\documentclass[11pt]{article}

\usepackage{fullpage}

\usepackage{amsmath,amsfonts,amsthm,mathrsfs,mathpazo,xspace,hyperref,graphicx,relsize}
\usepackage{endnotes,mathpazo}
\usepackage[svgnames]{xcolor}
\usepackage{bm}
\usepackage{times}
\usepackage{amssymb,latexsym}

\newtheorem{theorem}{Theorem}

\newtheorem{lemma}[theorem]{Lemma}

\newtheorem{claim}[theorem]{Claim}

\newtheorem{definition}[theorem]{Definition}

\newcommand{\ket}[1]{|#1\rangle}
\newcommand{\bra}[1]{\langle#1|}
\newcommand{\Tr}{\mbox{\rm Tr}}

\newcommand{\Es}[1]{\underset{#1}{\textsc{ E}}}

\newcommand{\C}{\ensuremath{\mathbb{C}}}

\newcommand{\beq}{\begin{eqnarray}}
\newcommand{\eeq}{\end{eqnarray}}

\newcommand{\Id}{\ensuremath{\mathop{\rm Id}\nolimits}}

\newcommand{\val}{\textsc{val}}
\newcommand{\valq}{\textsc{val}^{\!*}}
\newcommand{\valplus}{\textsc{val}_{\!+}}
\newcommand{\valqplus}{\textsc{val}_{\!+}^{\!*}}

\newcommand{\eps}{\varepsilon}

\newcommand{\poly}{\textrm{poly}}
\newcommand{\lin}{\mathcal{L}}

\newcommand{\card}[1]{\lvert#1\rvert}

\newcommand{\norm}[1]{\lVert#1\rVert}

\newcommand{\bignorm}[1]{\big\lVert#1\big\rVert}

\newcommand{\normi}[1]{\norm{#1}_\infty}

\makeatletter
\newcommand{\raisemath}[1]{\mathpalette{\raisem@th{#1}}}
\newcommand{\raisem@th}[3]{\raisebox{#1}{$#2#3$}}
\makeatother

\def\sqr{\!\raisemath{-1.0pt}{\Box}}
\def\sqrstr{\sqr\hspace*{-.35em}{\!\raisemath{-0.3pt}{*}}}
\newcommand{\qnorm}[1]{\norm{#1}_{\sqrstr}}
\newcommand{\sqrnorm}[1]{\norm{#1}_{\sqr}}

\newcommand{\plusnorm}[1]{\norm{#1}_{\!+}}
\newcommand{\bigqnorm}[1]{\bignorm{#1}_{\sqrstr}}
\newcommand{\rval}{\valqplus}
\newcommand{\iprod}[1]{\langle#1\rangle}

\definecolor{mygrey}{gray}{0.50}

\def\ot{{\otimes}}

\def\L{{\cal L}}

\def\mV{{\mathcal V}}
\def\mB{{\mathcal B}}
\def\mU{{\mathcal U}}
\def\mA{{\mathcal A}}

\begin{document}

\title{A parallel repetition theorem for entangled projection games}
\author{Irit Dinur\thanks{Department of Computer Science and Applied Math, The Weizmann Institute, Israel. Research supported by ERC grant number 239985. Part of the work was done while the author was visiting MIT supported by NSF Contract CCF-1018064, and by Simons Investigator Award of Shafi Goldwasser.} \and David Steurer\thanks{Departement of Computer Science, Cornell University. Part of this work was done at Microsoft Research New England.} \and Thomas Vidick\thanks{Newton Institute, Cambridge UK and Centre for Quantum Technologies, NUS Singapore. Partially supported by the Ministry of Education, Singapore under the Tier 3 grant MOE2012-T3-1-009. Part of this work was completed while the author was at MIT, supported by the National Science Foundation under Grant No. 0844626. }}
\date{}
\maketitle

\begin{abstract}
We study the behavior of the entangled value of two-player one-round projection games under parallel repetition. We show that for any projection game $G$ of entangled value $1-\eps < 1$, the value of the $k$-fold repetition of $G$ goes to zero as $O((1-\eps^c)^k)$, for some universal constant $c\geq 1$. If furthermore the constraint graph of $G$ is expanding we obtain the optimal $c=1$. Previously exponential decay of the entangled value under parallel repetition was only known for the case of XOR and unique games.
To prove the theorem we extend an analytical framework introduced by Dinur and Steurer for the study of the classical value of projection games under parallel repetition. Our proof, as theirs, relies on the introduction of a simple relaxation of the entangled value that is perfectly multiplicative. The main technical component of the proof consists in showing that the relaxed value remains tightly connected to the entangled value, thereby establishing the parallel repetition theorem. More generally, we obtain results on the behavior of the entangled value under products of arbitrary (not necessarily identical) projection games.

Relating our relaxed value to the entangled value is done by giving an algorithm for converting a relaxed variant of quantum strategies that we call ``vector quantum strategy'' to a quantum strategy. The algorithm is considerably simpler in case the bipartite distribution of questions in the game has good expansion properties. When this is not the case, the algorithm relies on a quantum analogue of Holenstein's correlated sampling lemma which may be of independent interest. Our ``quantum correlated sampling lemma'' generalizes results of van Dam and Hayden on universal embezzlement to the following approximate scenario: two non-communicating parties, given classical descriptions of bipartite states $\ket{\psi}$, $\ket{\varphi}$ respectively such that $\ket{\psi}\approx \ket{\varphi}$, are able to locally generate a joint entangled state $\ket{\Psi}\approx\ket{\psi}\approx\ket{\varphi}$ using an initial entangled state that is independent of their inputs.

\end{abstract}
\newpage

\section{Introduction}\label{sec:intro}

Two-player one-round games arise naturally in many areas of theoretical computer science. They are prominent in complexity theory, where they are a powerful tool for the study of constraint satisfaction problems, and in cryptography, where they give a polyvalent abstraction used to establish the security of many two-party primitives. They have also recently proven a very convenient framework for the study of some of the deepest issues in quantum mechanics, giving a novel viewpoint on the decades-old study of \emph{Bell inequalities}~\cite{Brunner14survey}, which are linear inequalities that must be satisfied by any family of distributions that can be generated locally according to the laws of classical mechanics, but can be violated if the distributions are allowed to be generated using quantum entanglement.

A game $G$ is specified by finite sets $\mathcal{U},\mathcal{V}$ of questions, $\mathcal{A}$, $\mathcal{B}$ of answers, a probability distribution $\mu$ on pairs of questions $(u,v)\in\mathcal{U}\times\mathcal{V}$, and an acceptance criterion $V\subseteq \mathcal{A}\times\mathcal{B}\times\mathcal{U}\times\mathcal{V}$ which states, for every possible pair of questions $(u,v)$, which pairs of answers $(a,b)\in\mathcal{A}\times\mathcal{B}$ are valid. The most basic quantity associated to a game is its \emph{value}. This can be defined operationally as the maximum success probability of two cooperating, but spatially isolated, players in the following game: a trusted party (the ``referee'') selects a pair of questions $(u,v)$ according to $\pi$, and sends $u$ to the first player (``Alice'') and $v$ to the second (``Bob''). Each player replies with an answer $a,b$, and the players win the game if and only if $V(a,b,u,v)=1$.

Remarkably, the precise definition of the value depends on the physical theory used to model the a priori vague assumption that the players be ``spatially isolated''. Under classical theory, isolated players are fully described by the (possibly randomized) functions they each apply to their respective question in order to determine their answer, and this interpretation leads to the \emph{classical value} $\val$ of the game. In contrast, in quantum theory isolated players are allowed any set of strategies that can be implemented by performing local measurements on a shared entangled state. The resulting value is called the \emph{entangled value} and denoted $\valq(G)$. Clearly for every game it holds that $\val\leq \valq$, and it is the discovery of Einstein, Podolsky and Rosen~\cite{epr} (formalized by Bell~\cite{Bell:64a}, simplified by Clauser et al.~\cite{Clauser:69a} and experimentally verified by Aspect et al.~\cite{Aspect81}) that there exist games for which the inequality is strict; indeed there are families of games $(G_n)$ for which $\val(G_n)\to 0$ but $\val^*(G_n)=1$~\cite{Raz98,Arvind:02}. One can go even further and consider the \emph{non-signaling value} $\val^{ns}$, which corresponds to players allowed to reproduce any bipartite correlations that do not imply signaling. Here again $\valq\leq\val^{ns}$, and there are games, such as the CHSH game~\cite{Clauser:69a}, for which the inequality is strict.

\medskip

One of the most fundamental questions one may ask about two-player games is that of the behavior of the value under \emph{product}. Given games $G$ and $H$, their product $G\otimes H$ is defined as follows: the question and answer sets are the cartesian product of those from $G$ and $H$; the distribution on questions is the product of the distributions, and the acceptance criterion the \textsc{and} of those of $G$ and of $H$. How does the value of $G\otimes H$ relate to that of $G$ and $H$? While it is clear that each of the three values defined above satisfies $\val(G\otimes H) \geq \val(G)\val(H)$, the reverse inequality, although intuitive, does \emph{not} hold in general. In particular, simple constructions of games $G$ are known such that $0<\val(G\otimes G)=\val(G)<1$~\cite{FeiLov92STOC}; similar constructions exist for $\valq$~\cite{CleveSUU08xor} and $\val^{ns}$~\cite{KempeR10noparallel}.

In spite of these examples, one may still ask for the behavior of $\val(G^{\otimes k})$, for ``large'' values of $k$. This is known as the \emph{parallel repetition} question: given a game $G$ such that $\val(G)<1$, does there exist a $\psi:[0,1]\to [0,1]$ such that $\psi(x)<1$ whenever $x>0$ and $\val(G^{\otimes k}) \leq (\psi(1-\val(G))^k$? If so, what form does $\psi$ take? Can it be approximately linear in the vicinity of $x=0$? Answering this question is of importance for many of the applications of two-player games. In cryptography, parallel repetition is a basic primitive using which one may attempt to amplify the security guarantees of a given protocol; in the study of Bell inequalities it can be used e.g. to amplify gaps between the quantum and non-signaling values; in complexity theory it is an important tool for hardness amplification.

For the case of the classical value, a sequence of works~\cite{Verbitsky,Fei91SCT,Feige2000} over the course of a decade led to the breakthrough by Raz~\cite{Raz98}, who was the first to provide a positive answer for general games: Raz showed that one can always take $\psi(x)=(1-x^c)^{d/\log{|\mathcal{A}\times\mathcal{B}|}}$, where $c$, $d$ are universal constants. Subsequent work focused on obtaining the best possible value for $c$ (the best known for general games is $c=3$~\cite{Hol09}) and on removing the dependence on the size of the answer alphabet for specific classes of games~\cite{Rao08parallel,BarakRRRS09parallel,RazR12projection}. For the case of the no-signaling value, Holenstein showed one can always take $\psi(x)=1-Cx^2$ for some constant $C>0$~\cite{Hol09}.

In contrast, for the case of the entangled value in spite of its importance the question is very poorly understood. Strong results are known for some very special classes of games such as XOR games~\cite{CleveSUU08xor}, for which repetition is exact (one can take $\psi(x)=1-x$) and unique games~\cite{KempeR10noparallel} (for $\psi(x)=1-Cx^2$, where $C>0$ is a universal constant). However, both these results, as well as related results motivated by cryptographic applications~\cite{HaenggiR09}, rely on the formulation of the entangled value as a semidefinite program, a characterization that is not believed to extend to more general games. Additional results are known but they only apply to specific games often originating from cryptography~\cite{MPA11,TomamichelFKW13monogamy}. Prior to this work the most general results known came from~\cite{KV11parallel}, where it is shown that a specific type of repetition inspired by work of Feige and Kilian~\cite{Feige2000}, in which the original game is mixed with ``consistency'' and ``free'' games, reduces the entangled value at a polynomial rate: provided $\val^*(G)<1$, the value $\val^*(G^{FK-\otimes k})$ of $k$ ``Feige-Kilian'' repetitions of $G$ behaves as $((1-\valq(G))k)^{-c}$ for some small $c>0$. (See ``related work'' below for additional discussion of more recent results that appeared after the initial completion of this work.)

\medskip

A recent work of Dinur and Steurer~\cite{DinurS13parallel} introduces a new approach to the parallel repetition question, focused on the case of \emph{projection games}. A projection game is one in which the referee's acceptance criterion has a special form: for any pair of questions $(u,v)$, any answer $b$ from the second player determines at most one valid answer $a = \pi_{uv}(b)$ for the first player.
Projection games are among the most interesting and widely-studied type of games. In particular, any local constraint satisfaction problem can be made into a projection game as follows: one player is asked for an assignment to all variables appearing in a constraint chosen at random, and the other is asked for an assignment to one of its variables. This simple transformation easily generalizes to convert any two-player game $G$ into a projection game $G'$, while essentially preserving the value: $1-\val(G')=\Theta(1-\val(G))$ (see Claim~\ref{claim:projection}). In particular, if one is only interested in ``amplifying the gap'' between $\val(G)=1$ and $\val(G)<1$ one can first map $G$ to $G'$ and then consider the parallel repetition of $G'$ itself, and this justifies the predominant role played by projection games in classical complexity theory. This transformation, however, may \emph{decrease} the entangled value arbitrarily whenever the optimal strategy for the players requires the use of entanglement (though we show that it can never \emph{increase} the value by too much; see Claim~\ref{claim:projection} for precise bounds). Nevertheless, many of the games studied in quantum information, such as the CHSH game~\cite{Clauser:69a} or the Magic Square game~\cite{Arvind:02} are projection games.

The approach of~\cite{DinurS13parallel} is based on the introduction of a relaxation of the game value, denoted $\valplus$. This relaxation can be defined for any game (we give the definition in Section~\ref{sec:intro-framework} below), and it is perfectly multiplicative. Moreover, for the case of projection games $\val_+$ turns out to remain closely related to $\val$, thus leading to a parallel repetition theorem. Although such a theorem already follows from Raz's general result~\cite{Raz98}, this arguably simpler approach matches the best parameters currently known~\cite{Rao08parallel}, which are known to be optimal~\cite{Raz08}. In addition, it yields new results for repetitions of games with small value and the case of few repetitions, which has implications for the approximability of the \textsc{label cover} and \textsc{set cover} problems.

\subsection{Our results}

We extend the analytical framework introduced in~\cite{DinurS13parallel} to the case of the entangled value $\valq$. As a consequence we obtain the following main theorem on the parallel repetition of the entangled value of projection games.

\begin{theorem}\label{thm:main}
There exists constants $c,C>0$ such that the following holds. For any projection game $G$,
$$\valq(G^{\otimes k}) \,\leq\, \big(1-C(1-\valq(G))^c\big)^{k/2}.$$
\end{theorem}

Although we do not attempt to fully optimize the constant $c$, the value that come out of our proof is $c\leq 12$. For the case of expanding games (see definition in Section~\ref{sec:prelim-games}) we obtain the optimal $c=1$. 

Parallel repetition results for the classical value were originally motivated by the study of multi-prover interactive proofs~\cite{FortnowRS88onthe}, and our result is likewise applicable to the study of classes of multi-prover interactive proofs with entangled provers. Letting $\textrm{MIP}^{\textrm{er}}_{1,s}(k)$ denote the class of languages having $k$-prover $1$-round interactive proofs in which completeness $c=1$ holds with unentangled provers, but soundness $s$ holds even against provers allowed to share entanglement, Theorem~\ref{thm:main} implies that $\textrm{MIP}^{\textrm{er}}_{1,s}(2) = \textrm{MIP}^{\textrm{er}}_{1,2^{-\poly}}(2)$ for any $s<1-\poly^{-1}$. This is because any protocol in $\textrm{MIP}^{\textrm{er}}_{1,s}(2)$ can be put into a form where the verifier's test is a projection constraint by following the reduction already discussed above, and described in Claim~\ref{claim:projection}; this will preserve both perfect completeness (for classical strategies) and soundness bounded away from $1$ (for quantum strategies). Prior to our work it was not known how to amplify soundness to exponentially small without increasing the number of rounds of interaction. It follows from~\cite{IV12,Vidick13xor} that $\textrm{MIP}^{\textrm{er}}_{1,1-\poly^{-1}}(3) = \textrm{NEXP}$, but very little is known about the $2$-prover class $\textrm{MIP}^{\textrm{er}}_{1,s}(2)$.

We believe that our results should find applications to a much wider range of problems. Going beyond the application to the parallel repetition question, our main contribution is the development of a precise framework in which general questions about the behavior of the value under product can be studied. This framework constitutes a comprehensive extension of the one introduced in~\cite{DinurS13parallel} for the study of the classical value: as in~\cite{DinurS13parallel}, we introduce a relaxation $\rval$ of the entangled value, prove that it is perfectly multiplicative, and show that it remains closely related to $\valq$. We find it remarkable that the framework from~\cite{DinurS13parallel}, introduced in a purely classical context, would find such a direct extension to the case of the entangled value. We hope that the tools developed in this extension will find further applications to the proof of product theorems in areas ranging from cryptography to communication complexity. Even though at a technical level the setting can appear quite different, some of the ideas put forth here could also prove useful to further removed areas such as the additivity conjecture for the minimum output entropy of quantum channels~\cite{AmosovHW00additivity,HaydenW08additivity,Hastings09superadditivity}.

We turn to a more detailed explanation of our framework, hoping to highlight precisely those tools and ideas that may find further application.

\subsection{Proof sketch}\label{sec:intro-framework}

In order to explain our approach it is useful to first review the framework introduced in~\cite{DinurS13parallel} for the study of the classical value.

\paragraph{Classical strategies.}

The starting point in~\cite{DinurS13parallel} consists in viewing games as \emph{operators} acting on the space of \emph{strategies}. In this language a strategy is simply a vector $\ket{f}$ of non-negative reals indexed by pairs $(u,a)$ of possible questions and answers: $f(u,a)$ is the probability that the strategy provides answer $a$ to question $u$. To any game one can associate a matrix $G$ such that, formally, the success probability of strategies $(\ket{f},\ket{g})$ for the players equals the vector-matrix-vector product $\bra{f}G\ket{g}$. The value of the game is then the norm of $G$ when viewed as an operator on the appropriately normed spaces of strategies.

The first crucial step taken in~\cite{DinurS13parallel} consists in relaxing the value of a game $G$ to the value of a symmetrized version of the game, which we call the \emph{square} $G^\dagger G$ of the game (this notation will be made precise in Section~\ref{sec:prelim-games}); we will denote the latter value by $\sqrnorm G$. In the square of a game $G$, the referee first samples a question $u$ for the first player as in $G$. He then independently samples two questions $v$ and $v'$ for the second player according to the conditional distribution. The players in $G^\dagger G$ are sent $v$ and $v'$ respectively. They have to provide answers $b$ and $b'$ such that there exists an $a$ such that both $(a,b)$ is a valid answer to $(u,v)$ in $G$, and $(a,b')$ is a valid answer to $(u,v')$. Note that now $G^\dagger G$ is in general not a projection game, even if $G$ was. In particular, $G^\dagger G$ treats both players symmetrically, and it turns out that we may always assume that they both apply the same strategy.
For the special case of projection games it is not hard to show that the value of the game and that of its square are quadratically related:
\beq\label{eq:intro-square}
\val(G)^2 \leq \sqrnorm G \leq \val(G).
\eeq
Indeed, using the algebraic language introduced above, the first inequality follows from the Cauchy-Schwarz inequality and the second is an easy observation.

The second step consists in observing that the application of the operator corresponding to the product $G\ot H$, where $G$ and $H$ are arbitrary projection games, can be decomposed as a product $(G\ot I)\cdot (I\ot H)$. Starting with a strategy $\ket{f}$ for $G\otimes H$, the result of applying $(I\ot H)$ to $\ket{f}$ is a new vector which no longer satisfies the strict normalization requirements of strategies. Understanding the new normalization leads to a further relaxation of $\sqrnorm G$, denoted $\val_+(G)$, in which the optimization is performed over the appropriate notion of ``vector strategies'', which intuitively are vectors that can be obtained by applying game operators to strategies. With the correct definition, it is easy to show that
\beq\label{eq:intro-valplus}
 \sqrnorm{G\ot H}^2 \le \val_+(G)\cdot \sqrnorm{H}^2 .
\eeq
The third and last step, which constitutes most of the technical work in~\cite{DinurS13parallel}, consists in showing that $\val_+(G)$ is a good approximation to $\sqrnorm G$.  This is done using a \emph{rounding procedure}, by which a vector strategy associated with a large $\val_+$ is mapped back to an actual strategy for the square game that also has a high value, thus serving as a witness for the value $\sqrnorm{G}$ being large as well. Altogether we get a bound on the value of $G\ot H$ as a product of a bound on the value of $G$ and a bound on the value of $H$. Repeated application of~\eqref{eq:intro-valplus} then leads to the following chain of inequalities (where the last approximate equality hides a polynomial dependence)
\beq\label{eq:intro-seq}
 \val(G^{\otimes k})^2 \leq \sqrnorm {G^{\otimes k}}^2 \leq \val_+(G)\cdot \sqrnorm{G^{\ot k-1}}^2 \le \cdots \le \val_+(G)^{k}\approx \val(G)^{k},
\eeq
proving the parallel repetition theorem.

\paragraph{Quantum strategies.}

Our goal now is to extend the above sketch to the case of the entangled value $\valq$. There is good reason for optimism. In contrast to most classical proofs that appear in the study of classical two-player games (such as those that go into Dinur's proof of the PCP theorem~\cite{Dinur07pcp}, or earlier approaches to parallel repetition~\cite{Verbitsky,Feige2000,Raz98}), which are often information-theoretic or combinatorial in nature, the analytic (one could say linear-algebraic) framework introduced in~\cite{DinurS13parallel} seems much better suited a priori to an extension to the quantum domain. %
Indeed, quantum strategies themselves are objects that live in $d$-dimensional complex vector space: instead of a vector $\ket{f}$ of non-negative reals (giving the probability of answering $a$ to question $u$, for every possible $u$ and $a$), a strategy is now a vector $\ket{A}$ of $d$-dimensional positive semidefinite matrices $A_u^a$ that describe the measurement to be performed upon receiving any question $u$. The normalization condition is $\sum_a A_u^a=\Id$ for every $u$, a constraint dictated by the formalism of measurements in quantum mechanics. Note that taking $d=1$ we recover classical strategies; quantum mechanics allows $d$ to be arbitrarily large.

At an abstract level, going from the classical to the entangled value thus solely requires us to think of the game $G$ as an operator acting on a bigger space of strategies, ``enlarging'' the non-negative reals to the space of $d$-dimensional positive semidefinite matrices. This operation is easily realized by ``tensoring with identity'', $G\to G \otimes \Id_{\C^d}$. 

It remains to show how to extend each of the steps outlined above.
The first step consists in obtaining an analogue of~\eqref{eq:intro-square}. As in the classical case the second inequality is easy, and follows by observing that, if $\ket{A}$ is a quantum strategy in $G^\dagger G$ then $(G \otimes \Id)\ket{A}$ is a valid strategy for the first player in $G$ (this notation will be made precise in Section~\ref{sec:prelim-games}.) The first inequality in~\eqref{eq:intro-square} is slightly more subtle. Although it can be shown directly by applying a suitable matrix version of the Cauchy-Schwarz inequality, we note that it can also be proven using known properties of a widely used  construction in quantum information theory, the \emph{pretty-good measurement} (PGM)~\cite{HausladenW94pgm,hjsww:capacity}. As it turns out, the relaxation $\valq(G)\to\qnorm{G}$ precisely corresponds to replacing the first player's optimal choice of strategy in $G$ by a near-optimal choice obtained from the pretty-good-measurement derived from the post-measurement states, on the first player's space, that arise from the second player's measurements. As a consequence,~\eqref{eq:intro-square} extends verbatim:
\begin{equation}
\valq(G)^2 \leq \qnorm{G}^2 \leq \valq(G).\tag{1*}
\end{equation}
Next we need to find an appropriate notion of vector strategy and corresponding relaxed value $\valqplus$. Here we are helped by the ``operational'' interpretation of a vector strategy as the result of the application of a game operator to a strategy meant for the product of several games. With the suitable generalization of the definition of classical vector strategies (see Definition~\ref{def:fractional-strategy}) we also obtain an analogue of~\eqref{eq:intro-valplus} for $\valqplus$:
\begin{equation}
 \qnorm{G\ot H}^2 \le \valqplus(G)\cdot \qnorm{H}^2.\tag{2*}
\end{equation}
Even though this is not directly needed for our purposes, we note that $\valqplus$ itself is perfectly multiplicative (see Lemma~\ref{lem:mult} for the easy proof).

Finally, and most arduous, is to relate the relaxation $\val_+^*$ back to the value of the square game, $\qnorm{G}^2$. In the classical case this involves rounding vector to actual strategies. In the quantum case rounding has to be performed synchronously by the players, and will necessarily involve the use of an entangled state. Intuitively, upon receiving their respective questions in $G$ the players need to initialize themselves in an entangled state that corresponds to the post-measurement state that they would be in, conditioned on having given a particular pair of answers to a given pair of questions in the game $H$ from which the vector strategy is derived (recall that, informally, vector strategies are the result of applying a game operator to a strategy meant for the product of two or more distinct games).

In case the bipartite distribution of questions in the game $G$ has good expansion properties we can show that this conditioned state is roughly the same regardless of the respective questions received by each player in $G$, so there is a way for players to renormalize their measurements and proceed. For the non-expanding case the states can differ significantly from question to question. Nevertheless, we can show that based on their respective questions the players are
able to agree on classical descriptions of two close states $\ket{\psi}\approx \ket{\varphi}$ that they respectively wish to be in.

Since the questions are not known to the players a priori, they need to generate the appropriate entangled states ``on the spot'', from an initial shared entangled state that is independent from $\ket{\psi}$ and $\ket{\varphi}$. Our new ``quantum correlated sampling'' lemma allows the players to do just this: given classical descriptions of $\ket{\psi}\approx\ket{\varphi}$ respectively, they are able to generate a joint entangled state $\ket{\Psi}\approx \ket{\psi}\approx \ket{\varphi}$ from an initial shared universal ``embezzlement state''~\cite{DamH03embezzlement} independent of $\ket{\varphi}$ or $\ket{\psi}$, without any communication. The lemma can be seen as a quantum variant of Holenstein's correlated sampling lemma~\cite{Hol09}, as well as a ``robust'' extension of the results of van Dam and Hayden on universal embezzlement states~\cite{DamH03embezzlement}. We discuss this lemma and related works in more detail in Section~\ref{sec:correlated}.

All steps having been extended, we obtain a direct generalization of the chain of inequalities~\eqref{eq:intro-seq} to the case of entangled strategies:\footnote{We note however that the approximate equality $\val_+^*(G) \approx \valq(G)$ that we obtain in the quantum case, although it suffices for our application to parallel repetition, is weaker than the one from~\cite{DinurS13parallel}. In particular, it is probably not tight.}
\begin{equation}
 \valq(G^{\otimes k})^2 \leq  \qnorm{G^{\otimes k}}^2 \leq \valqplus(G)\cdot \qnorm{G^{\otimes k-1}}^2 \le \cdots \le \valqplus(G)^k  \approx (\valq(G))^{k}.\tag{3*}
\end{equation}

\subsection{Additional related work}

Although few general results are known, the question of the behavior of the entangled value of a two-player game or protocol under parallel repetition arises frequently. It plays an important role in recent results on device-independent quantum key distribution~\cite{HaenggiR09,MPA11} and related cryptographic primitives~\cite{TomamichelFKW13monogamy}. The latter work considers parallel repetition of a game with quantum messages, a setting which is also the focus of~\cite{cooney}. The approach of~\cite{cooney} builds upon~\cite{JungePPVW10}, who relate the (classical) value of a two-player one-round game to the norm of the game when viewed as a tensor on the space $\ell_\infty(\ell_1)\otimes \ell_\infty(\ell_1)$. This is similar to our starting point of viewing games as operators acting on strategies, except that it considers the game as a bilinear form rather than an operator; the two points of view are equivalent. This perspective enables the authors to leverage known results on the study of tensor norms in Banach space (resp. operator space) theory to derive results on the classical (resp. entangled) value. To the best of our knowledge this connection has not led to an alternative approach to proving parallel repetition for general classes of games, although partial results were obtained in~\cite{cooney} for the special case of the entangled value of rank-one quantum games.

After the completion of this work two new results established an exponential parallel repetition theorem for two-player one-round games with entangled players in which the distribution on questions  is a product distribution. In~\cite{ChaillouxS13parallel} it is shown that the entangled value of games in which the distribution on questions is uniform decreases as
$$\valq(G^{\otimes k}) \leq (1-(1-\valq(G))^2)^{\Omega(k/\log|\mathcal{U}||\mathcal{V}||\mathcal{A}||\mathcal{B}|)}.$$
Very recently Jain et al.~\cite{JainPY13parallel} extended the result to arbitrary product distributions on the questions, while also removing the dependence on the number of questions: they obtained the bound
$$\valq(G^{\otimes k}) \leq (1-(1-\valq(G))^3)^{\Omega(k/\log|\mathcal{A}||\mathcal{B}|)}.$$
Both results are based on the use of information-theoretic techniques. They are incomparable to ours, as they apply to games in which the acceptance predicate is general but the input distribution is required to be product. In addition, both bounds above have a dependence on the number of answers in the game; while for the case of the classical value such a dependence is necessary~\cite{FeigeV02parallel}, for the entangled value it is not yet known whether it can be avoided.

\subsection{Open questions}

We briefly mention several interesting open questions. There still does not exist any parallel repetition result that applies to the entangled value of general, non-projection two-player one-round games, and it would be interesting to investigate whether our techniques could lead to (even relatively weak) results in the general setting. The case of three players is also of interest, and no non-trivial parallel repetition results are known either in the classical or quantum setting. In fact, the closely related question of XOR repetition of three-player games is known to fail dramatically even for the classical value~\cite{BBLV12}.

\paragraph{Organization of the paper.} We start with some important preliminaries in Section~\ref{sec:prelim}. There we introduce the representation of games and strategies that is used throughout the remainder of the paper. In Section~\ref{sec:relaxations} we introduce the two relaxations of the entangled value sketched in the introduction and give a more detailed overview of our proof. In Section~\ref{sec:approx} we prove the main technical component of our work, the relation between $\valqplus$ and $\qnorm{\cdot}^2$. Finally, in Section~\ref{sec:correlated} we state and prove the quantum correlated sampling lemma.

\paragraph{Acknowledgments.} We thank Attila Pereszl\'enyi for comments on an earlier version of this manuscript.

\section{Preliminaries}\label{sec:prelim}

\subsection{Notation}

We identify $\lin(\C^{d'},\C^{d})$, the set of linear operators from $\C^{d'}$ to $\C^{d}$, with the set of $d\times d'$ matrices with complex entries: if $X\in \lin(\C^{d'},\C^{d})$ then its matrix has entries $X_{a,b} = \bra{a}X\ket{b}$, where $\ket{a}$, $\ket{b}$ range over the canonical bases for $\C^d$, $\C^{d'}$ respectively, and we use the bra-ket notation to denote column vectors $\ket{b}$ and row vectors $\bra{a} = (\ket{a})^\dagger$, where $\dagger$ denotes the conjugate-transpose. We also write $\lin(\C^d)$ for $\lin(\C^d,\C^{d})$. The space $\lin(\C^{d'},\C^{d})$ is a Hilbert space for the inner product $\langle A,B\rangle := \Tr(A^\dagger B)$. We let $\|X\|_\infty$ be the operator norm of $X$, its largest singular value. A state $\ket{\Psi}\in \C^d$ is a vector with norm $1$.

The following simple calculation, sometimes known as Ando's identity, will be useful.

\begin{claim}\label{claim:ando} Let $X\in\lin(\C^d)$, $Y\in\lin(\C^{d'})$ be two operators and $\ket{\Psi}\in\C^d\otimes\C^{d'}$ a bipartite state with Schmidt decomposition $\ket{\Psi} = \sum_i \lambda_i \ket{u_i}\ket{v_i}$, where the $\lambda_i$ are non-negative reals. Then
\begin{equation}\label{eq:ando}
\bra{\Psi} X\otimes Y \ket{\Psi} \,=\, \Tr\big(XKY^TK^\dagger\big),
\end{equation}
where $K = \sum_i \lambda_i \ket{u_i}\bra{v_i}$ and the transpose is taken in the bases specified by the $\ket{u_i}$ and $\ket{v_j}$. In particular, if $\ket{u_i}=\ket{v_i}$ for every $i$, $K$ is positive semidefinite and~\eqref{eq:ando} evaluates to $\Tr(XKY^TK)$.
\end{claim}

\begin{proof}
The proof follows by direct calculation, expanding the left-hand side of~\eqref{eq:ando} using the Schmidt decomposition of $\ket{\Psi}$ and the right-hand side using the definition of $K$.
\end{proof}

We state a matrix analogue of the Cauchy-Schwarz inequality; we include a proof for completeness (see also~\cite[p.123]{pisierbook}).

\begin{claim}\label{claim:haagerup}
For any $d$ and operators $A_i\in\lin(\C^d)$, $B_i\in\lin(\C^{d'})$,
$$ \Big\|\sum_i \overline{A_i} \otimes B_i \Big\|_\infty^2 \leq \Big\|\sum_i \overline{A_i} \otimes A_i \Big\|_\infty\Big\|\sum_i \overline{B_i} \otimes B_i \Big\|_\infty.
$$
\end{claim}

\begin{proof}
Let $\ket{\Psi},\ket{\Phi}\in \C^d\otimes \C^d$ be unit vectors with Schmidt decomposition $\ket{\Psi} = \sum_i \lambda_i \ket{u_i}\ket{v_i}$ and $\ket{\Phi} = \sum_i \mu_i \ket{t_i}\ket{w_i}$. For any $A\in\lin(\C^d)$ and $B\in\lin(\C^{d'})$,
\begin{align*}
\bra{\Psi} \overline{A} \otimes B\ket{\Phi}&=\sum_{i,j} \lambda_i \mu_j \bra{u_i} \overline{A} \ket{t_j} \bra{v_i} B \ket{w_j}\\
&\leq \Big|\sum_{i,j}  \lambda_i \mu_j |\bra{u_i} \overline{A} \ket{t_j}|^2\Big|^{1/2}\Big|\sum_{i,j}  \lambda_i \mu_j |\bra{v_i} B \ket{w_j}|^2\Big|^{1/2}\\
&= \big|  \bra{\Psi_L} \overline{A}\otimes A \ket{\Phi_L} \big|^{1/2} \big|  \bra{\Psi_R} \overline{B}\otimes B \ket{\Phi_R} \big|^{1/2},
\end{align*}
where $\ket{\Psi_L} = \sum_i \lambda_i \ket{u_i}\overline{\ket{u_i}}$, $\ket{\Phi_L} = \sum_j \mu_j \ket{t_j}\overline{\ket{t_j}}$, $\ket{\Psi_R} = \sum_i \lambda_i \overline{\ket{v_i}}\ket{v_i}$ and $\ket{\Phi_R} = \sum_j \mu_j \overline{\ket{w_j}}\ket{w_j}$.
Applying the Cauchy-Schwarz inequality once more,
\begin{align}
\Big|\bra{\Psi} \Big(\sum_i \overline{A_i} \otimes B_i\Big) \ket{\Phi}\Big|
&\leq \Big| \bra{\Psi_L} \Big(\sum_i \overline{A_i} \otimes A_i \Big)\ket{\Phi_L} \Big|^{1/2}\Big| \bra{\Psi_R} \Big( \sum_i \overline{B_i}\otimes B_i\Big) \ket{\Phi_R}\Big|^{1/2}.\label{eq:haagerup-2}
\end{align}
Since~\eqref{eq:haagerup-2} holds for any $\ket{\Psi}$ and $\ket{\Phi}$, the claim is proved.
 \end{proof}

\subsection{Games and strategies}\label{sec:prelim-games}

\paragraph{Definitions.}

A two-player game is specified by finite question sets $\mathcal{U}$ and $\mathcal{V}$, finite answer sets $\mathcal{A}$ and $\mathcal{B}$, a distribution $\mu$ on $\mathcal{U}\times\mathcal{V}$, and an acceptance criterion $V \subseteq \mathcal{A}\times\mathcal{B}\times \mathcal{U}\times\mathcal{V}$. We also write $V(a,b,u,v)=1$ for $(a,b,u,v)\in V$. The game may also be thought of as a bipartite constraint graph, with vertex sets $\mathcal{U}$ and $\mathcal{V}$, edge weights $\mu(u,v)$, and constraints $V(a,b,u,v)=1$ on each edge $(u,v)$. We will write $\mu_L$ for the marginal distribution of $\mu$ on $\mathcal{U}$, and $\mu_R$ its marginal on $\mathcal{V}$. (We omit the subscripts $L$ and $R$ when they are clear from context.) We also often write $v\sim u$ to mean that $v$ is distributed according to the conditional distribution $\mu(v|u) = \mu(u,v)/\mu_L(u)$. The size of $G$ is defined as $|\mathcal{U}||\mV||\mathcal{A}||\mB|$.

In this paper we focus on projection games, which are games for which the acceptance criterion $V$ is such that for every $(u,v,b)\in \mathcal{U}\times\mathcal{V}\times\mathcal{B}$ there is at most one $a\in\mathcal{A}$ such that $V(a,b,u,v)=1$. Equivalently, for every edge $(u,v)$ the associated constraint is a \emph{projection} constraint $\pi_{u,v}:\mathcal{B}\to \mathcal{A}$ such that $\pi_{u,v}(b)$ is the unique $a$ such that $V(a,b,u,v)=1$ if it exists, and a special ``fail'' symbol $\perp$ otherwise. When the edge $(u,v)$ is clear from context we will write $b\to a$ to mean that $\pi_{uv}(b)=a$. We also write $b\leftrightarrow b'$ to mean that there exists an $a$ such that $b\to a$ and $b'\to a$.

Given a projection game $G$, let $H$ be the weighted adjacency matrix associated with the square of $G$: $H$ is the $|\mathcal{V}|\times |\mathcal{V}|$ matrix whose $(v,v')$-th entry equals $\mu(v,v') := \sum_u \mu(u)\mu(v|u)\mu(v'|u)$. Let $D$ be the diagonal matrix with the degrees $\mu_R(v)$ on the diagonal, and $L := \Id - D^{-1/2}HD^{-1/2}$ the normalized Laplacian associated with the square of $G$. We say that a family of games $(G_n)$, where $G_n$ has size $n$, is expanding if the second smallest eigenvalue of $L_n=L(G_n)$ is at least a positive constant independent of $n$.

\paragraph{Projection games as operators.}

Let $G$ be a two-player projection game.
We will think of $G$ as a linear operator $G:\C^{|\mathcal{V}|} \otimes \C^{|\mathcal{B}|}\to \C^{|\mathcal{U}|} \otimes \C^{|\mathcal{A}|}$ defined as follows:
$$ G \,:=\, \sum_{u,v} \,\mu(v|u) \sum_{a,\,b\to a}\, \ket{u}\bra{v}\otimes \ket{a}\bra{b} \in \mathcal{L}(\C^{|\mathcal{V}|}\otimes \C^{|\mathcal{B}|},\C^{|\mathcal{U}|}\otimes \C^{|\mathcal{A}|}).$$
In other words, for $\ket{B}\in \C^{|\mathcal{V}|} \otimes \C^{|\mathcal{B}|}$, let $B_v^b = \bra{v,b}B\rangle$ denote the value of $\ket{B}$ at the coordinates indicated by basis vectors $\ket{v}\in\C^{\mathcal{V}}$ and $\ket{b}\in\C^{\mathcal{B}}$. Then
$$(GB)_u^a \,:=\, \bra{u,a} G \ket{B} \,=\, \sum_v \mu(v|u) \sum_{b\to a} B_v^b.$$
Note that here we adopted the convention that questions $u\in\mathcal{U}$ are summed over, whereas questions $v\in\mathcal{V}$ are weighted by the corresponding conditional probability $\mu(v|u)$.

\paragraph{Classical strategies.}

The actions of players in a game $G$ give rise to a ``probabilistic assignment'', a collection of probability distributions $\{p(a,b|u,v)\}$ such that, for any pair of questions $(u,v)$, $p(\cdot,\cdot|u,v)$ is a probability distribution on pairs of answers to those questions. We may also represent $p$ as the rectangular $ |\mathcal{V}||\mathcal{B}|\times |\mathcal{U}||\mathcal{A}|$ matrix whose $((v,b),(u,a))$-th entry is $p(a,b|u,v)$.
The \emph{value} achieved by $p$ in the game is defined as
$$\val(G,p) \,:=\,  \sum_u \,\mu(u)  \,\sum_v \,\mu(v|u) \sum_a\,\sum_{b\to a} \,p(a,b|u,v)\,=\,\Tr_\mu( Gp ) ,$$
where we introduced a trace $\Tr_\mu$ on the set of all $X\in \mathcal{L}(\C^{|\mathcal{U}|} \otimes \C^{|\mathcal{A}|})$ by defining
$$ \Tr_\mu(X) \,:=\, \sum_u \mu(u) \sum_a X_{(u,a),(u,a)}.$$
In cases of interest the family of distributions $\{p(a,b|u,v)\}$ is not arbitrary, but has a bipartite structure which reflects the bipartite nature of the game. \emph{Classical} deterministic\footnote{Randomized strategies are convex combinations of deterministic strategies, thus a randomized strategy can always be replaced by a deterministic one achieving at least as high a value.} strategies correspond to the case when $p(a,b|u,v)=f(a|u)g(b|v)$ for functions $f(\cdot|u):\mathcal{A}\to\{0,1\}$ and $g(\cdot|v):\mathcal{B}\to\{0,1\}$ taking the value $1$ exactly once. The functions $f$ and $g$ may be represented as vectors
$$\ket{f} = \sum_{u,a} f(a|u)\ket{u}\ket{a} \in \C^{|\mathcal{U}|} \otimes \C^{|\mathcal{A}|}\qquad \text{and} \qquad \ket{g} = \sum_{v,b} g(b|v)\ket{v}\ket{b} \in \C^{|\mathcal{V}|} \otimes \C^{|\mathcal{B}|}$$
respectively. $p$ is then the rank-one matrix $p=\ket{g}\bra{f}$, and we may express the value as
$$ \val(G,p) =  \Tr_\mu( G p) =  \langle f,Gg \rangle_{\mu_L} =   \sum_u \mu_L(u)\sum_v \mu(v|u) \sum_a \sum_{b\to a} f(a|u)g(b|v),$$
where the inner product $\langle \cdot,\cdot \rangle_{\mu_L}$ is defined on $(\C^{\mathcal{U}}\otimes \C^{\mathcal{A}})\times(\C^{\mathcal{U}}\otimes \C^{\mathcal{A}})$ by
$$ \langle f,g\rangle_{\mu_L} := \sum_u \mu_L(u) \sum_a f(a|u)g(a|u).$$
We may similarly define an inner product $\langle \cdot,\cdot \rangle_{\mu_R}$ on $(\C^{\mathcal{V}}\otimes \C^{\mathcal{B}})\times (\C^{\mathcal{V}}\otimes \C^{\mathcal{B}})$, and we will omit the subscripts $L,R$ when they are clear from context. Given a game matrix $G$, we define its adjoint $G^\dagger$ as the unique matrix such that $\langle f,Gg\rangle_{\mu_L} = \langle G^\dagger f,g\rangle_{\mu_R}$ for all $f\in \C^{\mathcal{U}\times\mathcal{A}}$ and $g\in\C^{\mathcal{V}\times\mathcal{B}}$. Formally, if $G =\sum_{u,v} \mu(v|u) \sum_{a,b\to a} \ket{u}\bra{v}\otimes \ket{a}\bra{b}$ then $G^\dagger = \sum_{u,v} \mu(u|v) \sum_{a,b\to a} \ket{v}\bra{u}\otimes \ket{b}\bra{a}$.

\paragraph{Quantum strategies.}

Next we consider quantum strategies. A quantum strategy is specified by measurements $\{A_u^a\}_a$ for every $u$ and $\{B_v^b\}_b$ for every $v$, where in general a measurement is any collection of positive semidefinite operators, of arbitrary finite dimension $d$, that sum to identity. For any state $\ket{\Psi}$ representing the entanglement between the players,\footnote{In the literature the state $\ket{\Psi}$ is usually considered to be an integral part of the strategy. However it will be more convenient for us to not fix it a priori. Given measurement operators for both players in a game, it is always clear what is the optimal choice of entangled state; it is obtained as the largest eigenvector of a given operator depending on the game and the measurements (see below).} this strategy gives rise to the family of distributions
$$p_{\ket{\Psi}}(a,b|u,v)\,:=\,\bra{\Psi} \overline{A_u^a} \otimes B_v^b \ket{\Psi}.\footnote{The complex conjugate on $A$ is not necessary, but for our purposes it is natural to include it in light of the proof of Lemma~\ref{lem:best-response}.}$$
This formula, dictated by the laws of quantum mechanics, corresponds to the probability that the players obtain outcomes $a$, $b$ when performing the measurements $\{\overline{A_u^a}\}$, $\{B_v^b\}$ on their respective share of $\ket{\Psi}$. One can check that positive semidefiniteness of the measurement operators together with the ``sum to identity'' condition imply that $p_{\ket{\Psi}}(\cdot,\cdot|u,v)$ is a well-defined probability distribution on $\mathcal{A}\times \mB$.
To a quantum strategy we associate vectors
$$\ket{A} = \sum_{u,a} \ket{u}\ket{a}\otimes A_u^a \in \C^{|\mathcal{U}|}\otimes \C^{|\mathcal{A}|} \otimes \mathcal{L}(\C^d)\qquad\text{and}\qquad \ket{B} = \sum_{v,b} \ket{v}\ket{b}\otimes B_v^b\in  \C^{|\mathcal{V}|}\otimes \C^{|\mathcal{B}|} \otimes \mathcal{L}(\C^d).$$
(Note that these definitions reduce to classical strategies whenever $d=1$.)
 To express the success probability of this strategy in a game $G$ we extend the definition of the inner product $\langle \cdot,\cdot\rangle_\mu$ as follows.
\begin{definition}[Extended Inner Product]\label{def:innerprod}
The extended inner product
$$\langle \cdot,\cdot\rangle_{\mu_L} :\,\C^{|\mathcal{U}|} \otimes \C^{|\mathcal{A}|}\otimes \mathcal{L}(\C^d)\times \C^{|\mathcal{U}|} \otimes \C^{|\mathcal{A}|}\otimes \mathcal{L}(\C^d)\to \lin(\C^d)\otimes \lin(\C^d)$$
is defined, for $\ket{A}= \sum_{u,a} \ket{u}\ket{a}\otimes A_u^a$ and $\ket{B}= \sum_{u,a} \ket{u}\ket{a}\otimes B_u^a$,\footnote{Note the definition depends on a fixed choice of basis for the spaces $\C^{|\mathcal{U}|}$ and $\C^{|\mathcal{A}|}$.} by
$$ \langle A,B \rangle_{\mu_L} := \sum_u \mu_L(u) \sum_a \overline{A_u^a}\otimes  B_u^a.$$
\end{definition}
%
With this definition the success probability of the strategy $(\ket{A},\ket{B})$ in $G$ can be expressed as
\begin{align*}
\valq(G,\ket{A},\ket{B})&:= \big\|\langle A,(G\otimes \Id) B\rangle_\mu\big\|_\infty \\
&= \Big\|\sum_{u,a} \,\mu(u) \overline{A_u^a} \otimes \Big( \sum_{v} \,\mu(v|u)\sum_{b\to a} B_v^b \Big)\Big\|_\infty\\
&= \Big\|\sum_{u,v} \,\mu(u,v) \sum_{a,b\to a}  \overline{A_u^a} \otimes B_v^b \Big\|_\infty\\
&= \max_{\ket{\Psi}\in\C^d\otimes \C^d,\|\ket{\Psi}\|=1}\, \sum_{u,v} \,\mu(u,v) \sum_{a,\,b\to a}  \bra{\Psi}\overline{A_u^a} \otimes B_v^b \ket{\Psi}.
\end{align*}
We also define the entangled value of the game, $\valq(G)$, to be the highest value achievable by any quantum strategy:
\begin{align}
\valq(G) &= \sup_{\ket{A},\ket{B}} \valq(G,\ket{A},\ket{B})\notag\\
& = \sup_{\ket{A},\ket{B}} \big\|\langle A,(G\otimes \Id) B \rangle_\mu\big\|_\infty \notag\\
&= \sup_{\{A_u^a\},\{B_v^b\},\ket{\Psi}}\, \sum_{u,v}\mu(u,v) \, \sum_{a,b\to a}\, \bra{\Psi} \overline{A_u^a} \otimes B_v^b \ket{\Psi}\notag\\
&=\sup_{\{A_u^a\},\{B_v^b\},\ket{\Psi}} \,\sum_u\,\mu(u) \sum_a\, \bra{\Psi}\overline{A_u^a}\otimes B_u^a\ket{\Psi},\label{eq:val-1}
\end{align}
where here we slightly abuse notation and denote
\beq\label{eq:game-op}
B_{u}^a \,:=\, (\bra{u}\bra{a}\otimes \Id)(G\otimes \Id)\ket{B}\,=\,\sum_v \mu(v|u) \sum_{b\to a} B_v^b.
\eeq
We note that in the above the supremum may in general not be attained as optimal strategies may require infinite dimensions. In this paper we always restrict ourselves to finite dimensional strategies.\footnote{Thus when we say that $(\ket{A},\ket{B})$ achieve the value of $G$ we really mean that $(\ket{A},\ket{B})$ are finite-dimensional strategies whose value in $G$ can be made arbitrarily close to the optimum; for clarity we ignore this simple technicality in the whole paper.}

It is well-known that any two-player game can be made into a projection game while essentially preserving its classical value. The following claim gives a partial extension of this fact to the case of the entangled value. 

\begin{claim}\label{claim:projection}
There exists a polynomial-time computable transformation mapping  any two-player one-round game $G$ to a projection game $G'$ such that the following hold:
$$1-\val(G') \,\leq\,  1-\val(G) \,\leq\, 2(1-\val(G')).$$
In particular, $\val(G')=1$ if and only if $\val(G)=1$, and $1-\val(G') = \Theta(1-\val(G))$.
Moreover, for the entangled value we have the weaker bound
$$\valq(G') \,\leq\, \sqrt{\frac{1+\valq(G)}{2}},$$
 which implies $1-\valq(G') = \Omega(1-\valq(G))$.
\end{claim}

\begin{proof}
Let $G$ be a game with (without loss of generality disjoint) question sets $\mU$, $\mV$, answer sets $\mA$, $\mB$, distribution on questions $\mu$ and acceptance predicate $V$. Let $G'$ be the projection game corresponding to the following scenario. The referee selects a pair of questions $(u,v)$ at random from $\mu$, which it sends to the second player, and then sends either $u$ or $v$ to the first player, each with probability $1/2$. Formally, $G'$ is defined by question sets $\mU'=\mU\cup\mV$, $\mV'=\mU\times \mV$, answer sets $\mA'=\mA\cup\mB$, $\mB'=\mA\times \mB$, and a distribution $\mu'$ given by $\mu'(u,(u,v))=\mu'(u,v)/2$, $\mu'(v,(u,v))=\mu'(u,v)/2$, and $0$ otherwise. For any $(u,v)$ and $(a,b)$ let $\pi_{u,(u,v)}$ be such that $\pi_{u,(u,v)}(a,b)=a$ and $\pi_{v,(u,v)}(a,b)=b$ if $V(a,b,u,v)=1$, and there is no valid answer for the first player if the second player's answers are such that $V(a,b,u,v)=0$.

Then clearly $G'$ is a projection game. Let $\ket{f},\ket{g}$ be classical deterministic strategies for the players such that $\val(G,\ket{f},\ket{g})=\val(G)$. Consider the strategy $(\ket{f'},\ket{g'})$ for $G'$ in which $\ket{f'}$ answers as $\ket{f}$ to questions $u\in\mU$ and as $\ket{g}$ to questions $v\in\mV$, and $\ket{g'}$ answers as $(\ket{f},\ket{g})$. Then whenever the strategy $(\ket{f},\ket{g})$ provides answers to a pair of questions $(u,v)$ that satisfy the predicate $V$ the strategy $(\ket{f'},\ket{g'})$ gives answers to both $(u,(u,v))$ and $(v,(u,v))$ that are accepted in $G'$, hence
$$\val(G')\geq \val(G',\ket{f'},\ket{g'}) \geq \val(G,\ket{f},\ket{g}) = \val(G). $$
Conversely, let $(\ket{f'},\ket{g'})$ be a strategy for $G'$ such that $\val(G')=\val(G',\ket{f'},\ket{g'})$. Decompose $\ket{f'}$ into a pair of strategies $\ket{f},\ket{g}$ in $G$, depending on whether the question is $u\in\mU$ or $v\in\mV$. The pair $(\ket{f},\ket{g})$ will give a rejected answer to a pair of questions $(u,v)$ only if $(\ket{f'},\ket{g'})$ gave a rejected answer to at least one of the questions $(u,(u,v))$ and $(v,(u,v))$ in $G'$. In the worst case the $(1-\val(G',\ket{f},\ket{g}))$ probability that $(\ket{f'},\ket{g'})$ provides rejected answers in $G'$ is, say, fully concentrated on questions of the form $(u,(u,v))$. Hence
$$\val(G) \geq \val(G,\ket{f},\ket{g}) \geq 1-2(1-\val(G',\ket{f'},\ket{g'})) = 1-2(1-\val(G')).$$
Finally, let $(\ket{A},\ket{B})$ be a pair of quantum strategies such that $\valq(G')=\valq(G',\ket{A},\ket{B})$. To $\ket{A}$ we unambiguously associate measurement operators $\{A_u^a\}_a$ for every $u\in \mU$, and $\{A_v^b\}_b$ for $v\in\mV$. Hence
\begin{align*}
\valq(G') &= \Big\|\Es{u\sim v}\, \frac{1}{2} \sum_{(a,b):V(a,b,u,v)=1}  \overline{A_u^a} \otimes B_{u,v}^{a,b} + \overline{A_v^b} \otimes B_{u,v}^{a,b} \Big\|_\infty\\
&\leq \Big\|\Es{u\sim v}\,\frac{1}{4} \sum_{(a,b):V(a,b,u,v)=1} \overline{(A_u^a+ A_v^b)} \otimes (A_u^a+ A_v^b) \Big\|_\infty^{1/2} \Big\|\Es{u\sim v}\sum_{(a,b):V(a,b,u,v)=1} \overline{B_{u,v}^{a,b}}\otimes B_{u,v}^{a,b}\Big\|_\infty^{1/2}\\
&\leq \Big(\frac{1}{2} + \frac{1}{2}\Big\|\Es{u\sim v} \sum_{(a,b):V(a,b,u,v)=1} \overline{A_u^a} \otimes A_v^b) \Big\|_\infty\Big)^{1/2},
\end{align*}
where the first inequality uses Claim~\ref{claim:ando} and the last uses the triangle inequality for the operator norm and the fact that $\|\sum_i X_i \otimes Y_i \|_\infty = \|\sum_i Y_i \otimes X_i\|_\infty$ for any $X_i,Y_i$ to bound the first term, and uses $\sum_{a,b} B_{u,v}^{a,b}\leq \Id$ for every $u,v$, which implies
$$\Big\|\Es{u\sim v}\sum_{(a,b):V(a,b,u,v)=1} \overline{B_{u,v}^{a,b}}\otimes B_{u,v}^{a,b}\Big\|_\infty \le \Big\|\Es{u\sim v}\sum_{(a,b):V(a,b,u,v)=1} \Id \otimes B_{u,v}^{a,b}\Big\|_\infty \le \normi{\Id\otimes \Id}=1,$$
to bound the second. 
Hence the pair of strategies $(\ket{A_{|\mU}},\ket{A_{|\mV}})$ for $G$ achieves a value at least
$$\valq(G)\,\geq\,\valq(G, \ket{A_{|\mU}},\ket{A_{|\mV}}) \,\geq\, 2\,\valq(G')^2-1,$$
as claimed.
\end{proof}

\section{Relaxations of the game value}\label{sec:relaxations}

In this section we introduce two relaxations of the entangled value $\valq(G)$ of a projection game $G$. Both are quantum analogues of relaxations in \cite{DinurS13parallel}, and are used in the same way. The first relaxation, denoted $\qnorm{G}$, is related to playing a ``squared'' version of $G$ with two players Bob and Bob' treated symmetrically. It is defined in Section~\ref{sec:star}, and is easily seen to give a good approximation to $\valq$, as shown in the following lemma (see Section~\ref{sec:star} for the proof):

\begin{lemma}\label{lem:best-response}
For any projection game $G$,
\begin{equation}\label{eq:star-approx}
\valq(G)^2 \le \qnorm{G}^2 \le \valq(G).
\end{equation}
\end{lemma}

The second relaxation, denoted $\rval(G)$, is defined in Section~\ref{sec:rval}. It will be proven to be a good approximation to $\qnorm G$ and thus to $\valq$, although this will require more work.

\begin{lemma}\label{lem:approx}
For any projection game $G$,
\begin{equation}\label{eq:approx}
 \qnorm{G}^2 \le \rval(G) \le 1-C(1-\qnorm{G}^2)^c,
\end{equation}
for some positive constants $C,c>0$.
\end{lemma}
The proof of Lemma~\ref{lem:approx} is given in Section~\ref{sec:approx}.
The definition of $\rval$ is motivated by the following multiplicative property.

\begin{lemma}\label{lem:mult}
For any two projection games $G$ and $H$,
\begin{equation}\label{eq:mult}
\qnorm{G\otimes H}^2 \,\le\, \rval(G) \cdot \qnorm{H}^2,
\end{equation}
and $\rval$ is perfectly multiplicative:
\begin{equation}\label{eq:rmult}
\rval(G\otimes H) \,=\, \rval(G) \cdot \rval(H).
\end{equation}
\end{lemma}
The proof of Lemma~\ref{lem:mult} is given in Section~\ref{sec:rval}.

With these three inequalities in hand we easily derive the parallel repetition theorem, Theorem~\ref{thm:main}, as follows. By repeated applications of (\ref{eq:mult}), followed by~\eqref{eq:approx}, we get
\[ \qnorm{G^{\otimes k}}^2  = \qnorm{G\otimes G^{\otimes k-1}}^2 \le \rval(G)\cdot \qnorm{G^{\otimes k-1}}^2 \le \cdots \le (\rval(G))^k.
\]
Combining with~\eqref{eq:star-approx} and~\eqref{eq:approx} we get
\[ \valq(G^{\otimes k})^2  \le \qnorm{G^{\otimes k}}^2\le (\rval(G))^k \le (   1-C(1-\qnorm{G}^2)^c   )^k \le
(\,1-C(1-\valq(G))^c\,)^k,
\]

where the last step follows from (\ref{eq:star-approx}) and the monotonicity of $x\mapsto 1-C(1-x)^c$ on $[0,1]$.

\subsection{The square norm}\label{sec:star}
\begin{definition}\label{def:sqnorm} For a game $G$ and a quantum strategy $\ket B$ write $\qnorm{G\otimes \Id\ket{B}} :=  ( \normi{\iprod{ G\otimes \Id B,G\otimes \Id B}_\mu })^{1/2}$ and define
$$\qnorm{G} \,:=\,  \sup_{\ket B} \qnorm{G\otimes \Id\ket{B}},$$
where the supremum is taken over all $d$ and quantum strategies $\ket B \in \C^{|\mV|}\otimes \C^{|\mB|} \otimes \lin(\C^d)$.
\end{definition}
We note that $\qnorm{\cdot}$ is clearly homogeneous and non-negative. Although we will not use it, one can check that $\qnorm{\cdot}$ is also definite, and hence a norm, by setting $B_v^b = \Id$ for every $v$ and any $b$ such that $(G^\dagger G)_{(v,b),(v,b)}\neq 0$ (when it exists, and for an arbitrary $b$ otherwise).

Lemma~\ref{lem:best-response} claims that $\qnorm{G\otimes \Id\ket{B}}$ gives a good approximation to the maximum success probability in the game, when Bob uses the strategy specified by $\ket{B}$. We give a self-contained proof of the lemma below, but before proceeding readers familiar with quantum information theory may find it interesting to note that a direct proof of the first inequality can be derived using known properties of the pretty-good measurement (PGM)~\cite{HausladenW94pgm,hjsww:capacity}. We briefly indicate how. Suppose Bob's strategy in $G$ is fixed to $\ket{B}$. Upon receiving her question $u$, Alice has to decide on an answer $a$. She knows that Bob will receive a question $v$ distributed according to $\mu(\cdot|u)$ and apply his measurement, obtaining an outcome $b$ and resulting in the post-measurement state $\Tr_{2}( \Id\otimes \sqrt{B_v^b} \ket{\Psi}\bra{\Psi} \Id\otimes \sqrt{B_v^b})$ on her system. From her point of view, Alice needs to provide an answer $a$ such that $\pi_{uv}(b)=a$. Only knowing $u$, her task thus amounts to optimally distinguishing between the collection of post-measurement states
$$\rho_u^a = \Es{v\sim u}\, \sum_{b\to a}\, \Tr_{2}\big( \Id\otimes \sqrt{B_v^b} \ket{\Psi}\bra{\Psi} \Id\otimes \sqrt{B_v^b}\big).$$
If, instead of applying the optimal distinguishing measurement, Alice applied the pretty-good measurement (PGM) derived from this family of states then it follows from~\cite{BarnumK02pgm} that the players' success probability would be at most quadratically worse than what it would be was Alice to apply the optimal measurement. Using the explicit form of the PGM one can verify that the resulting value exactly corresponds to $\qnorm{G\otimes\Id \ket{B}}^2$, which proves the first inequality in~\eqref{eq:star-approx}.

\begin{proof}[Proof of Lemma~\ref{lem:best-response}]
We prove the following inequality, from which~\eqref{eq:star-approx} follows by taking the supremum over all $\ket{B}$:
\begin{equation}\label{eq:best-response}
 \max_{\ket{A}} \valq(G,\ket{A},\ket{B})^2 \,\leq\, \qnorm{G\otimes \Id\ket{B}}^2 \,\leq\,\max_{\ket{A}} \valq(G,\ket{A},\ket{B}).
\end{equation}
For the second inequality, using that $G$ is a projection game we note that for any $d$-dimensional strategy $\ket{B}$ for the second player, $(G\otimes \Id)\ket{B}$ is a valid strategy for the first player, hence
 $$\qnorm{(G\otimes \Id)\ket{B}}^2\,=\, \|\langle G\otimes \Id B,G\otimes \Id B\rangle_\mu\|_\infty \,\leq \,\max_{\ket{A}} \|\langle A,(G\otimes \Id)B\rangle_\mu\|_\infty = \max_{\ket{A}} \valq(G,\ket{A},\ket{B}).$$
To show the first, we write the following:
\begin{align*}
\valq(G,\ket{A},\ket{B}) &= \|\langle A,(G\otimes \Id)B\rangle_\mu\|_\infty \\
&= \Big\| \sum_{u} \,\mu(u)\sum_a \,\overline{A_u^a} \otimes B_u^a \Big\|_\infty\\
&\leq \Big\| \sum_{u}\, \mu(u) \sum_{ a} \,\overline{A_u^a} \otimes A_u^a \Big\|_\infty^{1/2}\Big\| \sum_{u} \,\mu(u) \sum_{ a} \,\overline{B_u^a} \otimes B_u^a \Big\|_\infty^{1/2}\\
&\leq \qnorm{\ket{(G\otimes \Id)B}},
\end{align*}
where for the first inequality we used the matrix Cauchy-Schwarz inequality stated in Claim~\ref{claim:haagerup}, and the last inequality uses $\sum_a A_u^a\leq \Id$ for every $u$.
\end{proof}

\subsection{The relaxation $\rval(G)$}\label{sec:rval}

In order to motivate our definition of $\rval$, let us consider two projection games $G,H$ and any quantum strategy $\ket B$ for $G\otimes H$ that achieves the optimal value $\qnorm{G\ot H}^2$ in the square game.
Letting $\kappa := \qnorm{G\ot H} / \qnorm{H}$, we want to bound $\kappa$ by a quantity that depends on $G$ and not on $H$. Consider the factorization $G\otimes H = (G\otimes I)(I \otimes H)$ where $I$ is the identity operator on the question and answer spaces associated with the first (resp. second) player in $H$ (resp. $G$); note that $I$ can also be understood as a game in which the two players are asked the same question and win if and only if they return the same answer. The application of $G\otimes H$ thus gives rise to a two step process
\begin{displaymath}
\ket {A'} \stackrel{G\otimes I}\longleftarrow \ket {A} \stackrel{I \otimes
    H}\longleftarrow \ket B ,
\end{displaymath}
mapping $\ket B $ to $\ket {A} := ((I \otimes H) \otimes \Id)\ket B$ and then mapping $\ket A$ to $\ket {A'}:= ((G\ot I)\ot \Id)\ket A$. Let us view $\ket B$ as a table with rows indexed by $\mV_G\times \mB_G$  and columns indexed by $\mV_H\times \mB_H$, where $\mV_G,\mV_H$ and $\mB_G,\mB_H$ are the question and answer sets associated with the second player in $G$ and $H$ respectively, and whose entries are measurement operators, i.e. elements in $\L(\C^d)$. Then $\ket A$ is the result of applying $H\otimes \Id$ on each row of $\ket B$ separately, and we apply $G\ot \Id$ on each column of $\ket A$ separately to get $\ket{ A'} = (G\ot I\ot \Id)\ket A$.

It is instructive to view the strategy $\ket B$ as an assignment to each $v\in \mV_G$ and $b\in \mB_G$ of a row vector $(\bra v \bra b \ot I \ot \Id) \ket B$ of dimensions $\card {\mV_H}\card{\mB_H}$ (whose entries are again in $\L(\C^d)$). Observe that for any $v$, $\ket{B_v} = \sum_b (\bra v \bra b \ot I\ot \Id) \ket B$ is a quantum strategy for $H$, since for each question $v'$ for $H$, the sum over answers $b'$ of
$$(\bra{v'}\bra{b'}\ot \Id)\ket{B_v} \,=\, B_{v,v'}^{b'}\,=\, \sum_b \,B_{v,v'}^{b,b'}$$
 is $\sum_{b'} B_{v,v'}^{b'} = \sum_{b'} \sum_b B_{v,v'}^{b,b'} =\Id$. In particular, $\qnorm{ H\ot \Id\ket {B_v} }^2 \le \qnorm{H}^2$. We write \beq\label{eq:def-av}
\ket{A_v} \,:=\, \sum_b (\bra v \bra b \ot I \ot \Id) \ket A
\eeq
and observe that it is equal to $H\ot \Id\ket {B_v}$, hence it satisfies $\qnorm{\ket{A_v}} \le \qnorm{H}$ for every $v$.
Thus the ratio between $\qnorm{G\ot I\ot \Id\ket A}$ and $\max_v \qnorm {A_v}$ is at least $\kappa= \qnorm{G\ot H} / \qnorm H$. As a result of our observations the ratio $\kappa$ can be upper bounded in a manner that depends only on $G$ and is {\em independent of $H$}. Abstracting the set $\mathcal{U}_H\times \mathcal{A}_H$ associated with pairs of questions and answers for the first player in $H$ as $\Omega$ for some discrete set $\Omega$,\footnote{In order for the extended inner product $\langle \cdot,\cdot \rangle_\mu$ to remain well-defined, we also need to equip $\Omega$ with a measure -- here, it would be the cartesian product of the probability measure $\mu_L$ on $\mathcal{U}_H$ and the counting measure on $\mathcal{A}_H$.}  we are led to the definition of $\rval(G)$ as the supremum of $\qnorm{G\otimes I_{\Omega} \otimes \Id_{\C^d}\ket{A}}^2$ ranging over vector quantum strategies $\ket A$ with norm $\plusnorm{A}\le 1$ defined as follows.

\begin{definition}[Fractional Strategy and Vector Strategy]\label{def:fractional-strategy}
Let $G$ be a projection game and ${\Omega}$ a discrete measured space. An element
$$\ket{A} \,=\, \sum_{v,b} \ket{v}\ket{b} \otimes A_{v}^b \in \C^{|\mathcal{V}|} \otimes  \C^{|\mathcal{B}|}\otimes \lin(\C^d)$$
 is a \emph{fractional quantum strategy} for $G$ if for every $v,b$ the matrix $A_{v}^b$ is positive semidefinite and $A_{v}:=\sum_b A_{ v}^b \leq \Id$ for every $v$.
%
A \emph{vector quantum strategy} is an element
\[ \ket A = \sum_{\omega\in\Omega} \ket \omega \ket {A_\omega}\in \C^{|\Omega|} \ot \C^{|\mathcal{V}|} \otimes  \C^{|\mathcal{B}|}\otimes \lin(\C^d)
\]
such that each $\ket {A_\omega}$ is a fractional quantum strategy.
The \emph{norm} of a vector quantum strategy is defined as
\begin{equation}\label{eq:defnorm}
 \plusnorm{\ket{A}} \,:=\, (\max_v \big\|\Es{\omega} \, \overline{A_{\omega v}} \otimes A_{\omega v} \big\|_\infty )^{1/2}.
\end{equation}
\end{definition}

The definition of $\rval$ is given by,
\begin{definition}[The relaxation $\rval$]\label{def:valplus}
Let $G$ be a projection game. Then
\begin{align*}
 \rval(G)&:= \sup_{\Omega} \sup_{\substack{\ket{A}\in\C^{|\mathcal{V}|} \otimes  \C^{|\mathcal{B}|}\otimes \C^{|\Omega|}\otimes \lin(\C^d) \\\plusnorm{A}\leq 1}} \bigqnorm{G\otimes I_{\Omega} \otimes \Id_{\C^d}\ket{A}}^2,
\end{align*}
where the supremum is taken over all discrete measured spaces $\Omega$.
\end{definition}
With these definitions in place we prove Lemma~\ref{lem:mult} relating the square norm of a product of games to $\rval$.

\begin{proof}[Proof of Lemma~\ref{lem:mult}]
Let $\ket{B}$ be an optimal strategy in the square game associated to $G\otimes H$.
It follows immediately from our observations above that $\ket A = I\ot H\ot \Id \ket B$ is a vector quantum strategy for $G$ (where the space $\Omega = \mathcal{U}_H \times \mathcal{A}_H$, and the measure is the cartesian product of the probability measure $\mu_L$ on $\mathcal{U}_H$ and the counting measure on $\mathcal{A}_H$) whose norm is $\plusnorm{\ket{A}} \le \qnorm{H}$.
 This means that
\[ \qnorm{G\ot H }^2 = \qnorm{G\ot H\ot\Id \ket B}^2 =  
 \qnorm{G\ot I \ot \Id \ket A}^2 \le \rval(G) \cdot \qnorm{H}^2, \]
where the last inequality comes by observing that $\frac 1 {\qnorm{H}} \ket A$ is a vector strategy with norm $\plusnorm\cdot$ at most $1$, so its value is at most $\rval(G)$.

Multiplicativity of $\rval$ follows along the same lines. First we note that $\rval(G\otimes H)\geq \rval(G)\rval(H)$ is clear. To show the converse, proceed as above by first fixing an optimal vector quantum strategy $\ket{B}$ for the square game associated to $G\otimes H$, such that $\plusnorm{\ket{B}}=1$. As in the above, it is easy to see that $\ket A = I\ot H\ot \Id \ket B$ is a vector quantum strategy for $G$ whose norm satisfies $\plusnorm{\ket{A}} \le \rval(H)$. Thus
\[ \rval(G\ot H ) = \qnorm{G\ot H\ot\Id \ket B}^2 =
 \qnorm{G\ot I \ot \Id \ket A}^2 \le \rval(G) \cdot \rval(H), \]
proving the claim.
\end{proof}

\section{Relating $\rval(G)$ to the square norm}\label{sec:approx}

In this section we prove Lemma~\ref{lem:approx}, which states that $\rval(G)$ is a good relaxation of the square norm $\qnorm{G}$ of a projection game and establishes the last step in our proof of the parallel repetition theorem, Theorem~\ref{thm:main}. We will also show that if $G$ is an expanding projection game then one can take $c=1$ in the bound $\rval(G) \le 1-C(1-\qnorm{G}^2)^c$.


To prove the lemma, we need to show that the existence of a good vector strategy for the players Bob and Bob' in the square game $G^\dagger G$ implies that $\qnorm{G}^2$ is large, i.e. there also exists a good (standard) quantum strategy for the players Alice and Bob in $G$. We will establish this by describing an explicit rounding procedure mapping the former to the latter. The rounding argument is simpler in case $G$ has the additional property of being expanding (see Section~\ref{sec:prelim-games} for the definition), and we give the proof in that case in Section~\ref{sec:expand}. In Section~\ref{sec:non-expand} we treat the case of general projection games. In that case the rounding argument is more involved and relies on a ``quantum correlated sampling'' lemma which is stated and proved in Section~\ref{sec:correlated}.

In both cases, the starting point for the rounding procedure is the existence of a vector strategy $\ket{\hat{A}}$ and entangled state $\ket{\hat{\Psi}}$ satisfying inequality~\eqref{eq:aomega-val} in the following claim, which is essentially a restatement of the inequality ``$\rval(G) \geq 1-\eta$''.

\begin{claim}\label{claim:aomega}
Let $G$ be a projection game and $\eta>0$ such that $\rval(G) \geq 1-\eta$. Then there exists a discrete measured space $\Omega$, an integer $d$, a bipartite state $\ket{\hat{\Psi}}\in\C^d\otimes \C^d$ and a vector strategy $\ket{\hat{A}}\in \C^{|\Omega|}\otimes\C^{|\mathcal{V}|} \otimes  \C^{|\mathcal{B}|}\otimes \lin(\C^d)$ such that for every $\omega$ and $v,b$, $\hat{A}_{\omega v}^b \geq 0$ and $\hat{A}_{\omega v} = \sum_b \hat{A}_{\omega v}^b \leq \Id$, and
\beq\label{eq:aomega-val}
\Es{\omega}\, \Es{v\sim v'} \,\sum_{b\leftrightarrow b'} \,\bra{\hat{\Psi}} \overline{\hat{A}_{\omega v}^b} \otimes \hat{A}_{\omega v'}^{b'} \ket{\hat{\Psi}} \,\geq\, (1-\eta)  \,\max_v\big\{\,\Es{\omega}\,\bra{\hat{\Psi}} \overline{\hat{A}_{\omega v}} \otimes \hat{A}_{\omega v} \ket{\hat{\Psi}}\,\big\},
\eeq
where formally $\textsc{E}_{v\sim v'}\sum_{b\leftrightarrow b'}$ is shorthand for $ \sum_u \mu(u) \sum_a \sum_{v,v'} \mu(v|u)\mu(v'|u) \sum_{b\to a,b'\to a} $. Furthermore, without loss of generality $\ket{\hat{\Psi}}$ can be chosen so as to have the following symmetry: its reduced densities on either subsystem are identical, and denoting either by $\hat{\rho}$, for any $X,Y$ it holds that 
\begin{equation}\label{eq:ando-good}
\bra{\hat{\Psi}} X \otimes Y \ket{\hat{\Psi}} \,=\, \Tr\big(\overline{X} \hat{\rho}^{1/2} Y \hat{\rho}^{1/2}\big).
\end{equation}
\end{claim}

\begin{proof}
By definition of $\rval$, there exists a discrete measured space $\Omega$ and a vector strategy $\ket{\hat{A}}$ such that $\plusnorm{\ket{\hat{A}}}=1$ and $\qnorm{\Id_\Omega \otimes G \otimes \Id \ket{\hat{A}}}^2\geq 1-\eta$. Recalling the definition of $\plusnorm{\cdot}$ (see Definition~\ref{def:fractional-strategy}) and of $\qnorm{\cdot}$ (see Definition~\ref{def:sqnorm}), we may reformulate this statement as the inequality
\begin{equation}\label{eq:aomega-1}
 \Big\|\Es{\omega} \Es{v\sim v'} \sum_{b\leftrightarrow b'} \, \overline{\hat{A}_{\omega v}^b} \otimes \hat{A}_{\omega v'}^{b'}\Big\|_\infty \geq (1-\eta) \max_v \Big\| \Es{\omega} \, \overline{\hat{A}_{\omega v}} \otimes \hat{A}_{\omega v} \Big\|_\infty.
\end{equation}
Letting $\ket{\hat{\Psi}}$ be a state which achieves the operator norm on the left-hand side gives~\eqref{eq:aomega-val}. The fact that $\ket{\hat{\Psi}}$ can be assumed to take the claimed form follows from the symmetry of the left-hand side of~\eqref{eq:aomega-1}.
\end{proof}

Let $\ket{\hat{A}}$ be a vector strategy and $\ket{\hat{\Psi}}$ a state such that~\eqref{eq:aomega-val} holds. Our goal is to identify a quantum strategy $\ket{\tilde{A}}$ such that $\qnorm{G\otimes \Id \ket{\tilde{A}}}^2 \geq 1-O(\eta^{1/c})$, which by Claim~\ref{claim:aomega} will suffice to prove Lemma~\ref{lem:approx}. 
The ``rounding procedure'' constructing $\ket{\tilde{A}}$ will differ in the expanding and non-expanding cases. Both cases however build on the same measurement operators which we now define. 

Fix an arbitrary $\omega\in\Omega$. The only ``defect'' of~$\ket{\hat{A}_\omega}$ that prevents it from directly giving us a quantum strategy is that it is only a fractional strategy, meaning that for any question $v$ the sum $\hat{A}_{\omega v} = \sum_b \hat{A}_{\omega v}^b$ may not equal the identity. It is natural to define a re-normalized strategy as follows. Let $U_{\omega v}$ be a unitary such that
\begin{equation}\label{eq:defu}
U_{\omega v}A_{\omega v}^{1/2} \rho^{1/4} =  \rho^{1/4}A_{\omega v}^{1/2} U_{\omega v}^\dagger = \big(\rho^{1/4}A_{\omega v}\rho^{1/4} \big)^{1/2}
\end{equation}
 is Hermitian positive semidefinite; such a unitary can be obtained from the singular value decomposition of $A_{\omega v}^{1/2} \rho^{1/4}$. For every pair of questions
$v,v'\in\mathcal{V}$ we introduce the post-measurement state
\beq\label{eq:def-tildepsi}
\ket{\Psi_{\omega vv'}}\,:=\, \overline{U_{\omega v}} \overline{\hat{A}_{\omega v}}^{1/2} \otimes U_{\omega v'}\hat{A}_{\omega v'}^{1/2} \ket{\hat{\Psi}}.\
\eeq
The state $\ket{\Psi_{\omega vv'}}$ is the post-measurement state that corresponds to applying the binary measurements $\{\overline{\hat{A}_{\omega v}}, \Id-\overline{\hat{A}_{\omega v}}\}$ for the first player, $\{\hat{A}_{\omega v'},\Id-\hat{A}_{\omega v'}\}$ for the second, to $\ket{\hat{\Psi}}$ and conditioning on both of them obtaining the first outcome. In general the post-measurement state is only defined up to a local unitary, and this freedom is represented in the unitaries $U_{\omega v}$ and $U_{\omega v'}$; our particular choice of unitaries satisfying~\eqref{eq:defu} will prove convenient in the analysis. 
Next for every question $v\in\mathcal{V}$ and answer $b\in\mathcal{B}$ we define the measurement operator
\beq\label{eq:def-tildea}
\tilde{A}_{\omega v}^b:= U_{\omega v}\hat{A}_{ \omega v}^{-1/2} \hat{A}_{\omega v}^b \hat{A}_{ \omega v}^{-1/2} U_{\omega v}^\dagger,
\eeq
where here $\hat{A}_{\omega v}^{-1/2}$ denotes the square root of the pseudo-inverse of $\hat{A}_{\omega v}=\sum_b \hat{A}_{\omega v}^b$. Again, there is always a unitary degree of freedom in the choice of the square root, and the unitaries $U_{\omega v}$, the same as in~\eqref{eq:def-tildepsi}, represent that degree of freedom. With this definition it is easy to verify that each $\tilde{A}_{\omega v}^b$ is positive semidefinite and that $\sum_b \tilde{A}_{\omega v}^b \leq \Id$; since we may always add a ``dummy'' outcome in order for the measurement operators to sum to identity, $\{\tilde{A}_{\omega v}^b\}_b$ is easily extended into a well-defined measurement and $\ket{\tilde{A}} := \sum_{v,b} \ket{v,b}\otimes \tilde{A}_{\omega v}^b$ is a valid quantum strategy in $G^\dagger G$.

Now suppose that, upon receiving their respective questions $v$ and $v'$, players Bob and Bob' in $G^\dagger G$ were to measure their respective share of the (re-normalized) state $\ket{\Psi_{\omega vv'}}$ using the measurements given by the $\{\overline{\tilde{A}_{\omega v}^b}\}_b$, $\{\tilde{A}_{\omega v'}^{b'}\}_{b'}$ respectively. The probability that they obtain the pair of outcomes $(b,b')$ is given, up to normalization by $\|\ket{\Psi_{\omega vv'}}\|^{-2}$, by
\begin{align}
\bra{\Psi_{\omega vv'}} \overline{\tilde{A}_{\omega v}^b}\otimes \tilde{A}_{\omega v'}^{b'} \ket{\Psi_{\omega vv'}}&=\bra{\hat{\Psi}}  \big(\overline{\hat{A}_{\omega v}}^{1/2}\overline{U_{\omega v}}^\dagger \otimes \hat{A}_{\omega v'}^{1/2}U_{\omega v'}^\dagger\big) \big(\overline{U_{\omega v}\hat{A}_{ \omega v}^{-1/2} \hat{A}_{\omega v}^b \hat{A}_{ \omega v}^{-1/2} U_{\omega v}^\dagger}\notag\\
&\hskip2cm\otimes U_{\omega v'}\hat{A}_{ \omega v'}^{-1/2} \hat{A}_{\omega v'}^b \hat{A}_{ \omega v'}^{-1/2} U_{\omega v'}^\dagger\big) \big(\overline{U_{\omega v}} \overline{\hat{A}_{\omega v}}^{1/2} \otimes U_{\omega v'}\hat{A}_{\omega v'}^{1/2} \big)\ket{\hat{\Psi}}\notag\\
&= \bra{\hat{\Psi}} \overline{\hat{A}_{\omega v}^b}\otimes \hat{A}_{\omega v'}^{b'} \ket{\hat{\Psi}},\label{eq:tilde-good}
\end{align}
perfectly reproducing the correlations induced by the fractional strategy $\ket{\hat{A}_\omega}$ together with $\ket{\hat{\Psi}}$.
Thus if it were the case that for all $(v,v')$, $\ket{\Psi_{\omega vv'}} = \ket{\Psi_\omega}$, a vector independent of $(v,v')$, then the players could use $\ket{\Psi_\omega}$ (for an appropriate, ``good'' choice of $\omega$) as their initial shared entangled state and perfectly emulate $\ket{\hat{A}_\omega}$ using the quantum strategy $\ket{\tilde{A}}$. 


While it may unfortunately not be the case that the $\ket{\Psi_{\omega vv'}}$ are independent of $(v,v')$, the main claim in the proof of Lemma~\ref{lem:approx} will establish that they are close, on average, when $v$ and $v'$ are neighboring vertices in the constraint graph. In the case where the game is expanding this will be sufficient, as we will be able to conclude that all states $\ket{\Psi_{\omega vv'}}$ are close to a single $\ket{\Psi_{\omega}}$ independent of $v$ and $v'$. In the non-expanding case we will rely on a more complicated strategy that involves a step of correlated sampling in which the players, \emph{after} having received their respective $v$ and $v'$, jointly sample an $\omega$ and create the corresponding bipartite state $\ket{\Psi_{\omega vv'}}$ locally. 

We first turn to the case of expanding games, for which we can give a simpler (and tighter) analysis. 



\subsection{The expanding case}\label{sec:expand}

Suppose that $G$ is expanding. 
Our first step consists in fixing a ``good'' value $\omega\in \Omega$ and restricting our attention to the fractional strategy $\ket{\hat{A}_\omega} := (\bra{\omega}\otimes I \otimes \Id)\ket{\hat{A}}$ specified by the operators $\hat{A}_{\omega v}^b$ obtained from that $\omega$. Using that the max is larger than the average, Eq.~\eqref{eq:aomega-val} implies
\beq\label{eq:aomega-val-expand}
\Es{\omega}\,\Big( \Es{v\sim v'} \,\sum_{b\leftrightarrow b'} \,\bra{\hat{\Psi}} \overline{\hat{A}_{\omega v}^b} \otimes \hat{A}_{\omega v'}^{b'} \ket{\hat{\Psi}} \Big)\,\geq\, (1-\eta)  \,\Es{\omega}\,\Big(\Es{v}\,\bra{\hat{\Psi}} \overline{\hat{A}_{\omega v}} \otimes \hat{A}_{\omega v} \ket{\hat{\Psi}}\,\Big).
\eeq
For the remainder of this section fix an $\omega$ such that~\eqref{eq:aomega-val-expand} holds for that $\omega$. The only property we will need of the $\{\hat{A}_{\omega v}^b\}$ in order to construct a good strategy in $G^\dagger G$ is that they are positive semidefinite operators which satisfy that inequality. (In contrast, for the non-expanding case, Eq.~\eqref{eq:aomega-val-expand} by itself turns out to be too weak an inequality, and we must work with~\eqref{eq:aomega-val}.)

Having fixed a value for $\omega$, for clarity of notation for every $v,v'$ we let $U_v$ be the unitary defined by~\eqref{eq:defu}, $\ket{\Psi_{vv'}}$ the state defined in~\eqref{eq:def-tildepsi}, and $\tilde{A}_v^b$ the measurement operators introduced in~\eqref{eq:def-tildea}. Let 
\begin{equation}\label{eq:def-sigma}
\sigma :=  \Big(\Es{w}\,\big\|\ket{\Psi_{ww}}\big\|^{2}\Big)^{-1} \,\Es{w}\,\ket{\Psi_{ww}}\bra{\Psi_{ww}}.
\end{equation}
The operators $\tilde{A}_v^b$, together with the density matrix $\sigma$, form a well-defined strategy for the players in the square game. In order to prove Lemma~\ref{lem:approx} (for the case of expanding games) it remains to bound the error
\begin{equation}
\label{eq:def-eps}
\eps\,:=\,\Es{v\sim v'} \sum_{b\nleftrightarrow b'} \Tr\big( \overline{\tilde{A}_{v}^{b}}\otimes \tilde{A}_{v'}^{b'}\,\sigma\big),
\end{equation}
incurred by that strategy, under the assumption that~\eqref{eq:aomega-val-expand} holds. We show the following. 

\begin{claim}\label{claim:expand-bound}
Suppose~\eqref{eq:aomega-val-expand} holds, and the constraint graph $G$ is such that the smallest nonzero eigenvalue of the Laplacian $L := \sum_v \ket{v}\bra{v} - \sum_{v,v'\sim v} \mu(v,v')\mu(v)^{-1/2}\mu(v')^{-1/2} \ket{v'}\bra{v}$ is at least $\lambda>0$, where here $\mu(v,v')$ is the distribution on questions in the square game, as defined in Section~\ref{sec:prelim-games}. Let $(\tilde{A}_v^b,\sigma)$ be the strategy defined above. Then 
\begin{equation}\label{eq:sigma-good}
 \eps\,=\, \Es{v\sim v'} \sum_{b\nleftrightarrow b'} \Tr\big( \overline{\tilde{A}_{v}^{b}}\otimes \tilde{A}_{v'}^{b'}\,\sigma\big) \,=\,O\big(\eta/\lambda\big).
\end{equation}
\end{claim} 

Before proceeding with the proof of Claim~\ref{claim:expand-bound} we show that it implies Lemma~\ref{lem:approx}.

\begin{proof}[Proof of Lemma~\ref{lem:approx}, expanding case]
Let $\eta>0$ be such that $\rval(G) \geq 1-\eta$. Then it follows directly from Claim~\ref{claim:aomega} that~\eqref{eq:aomega-val-expand} holds for this choice of $\eta$. Let $(\tilde{A}_v^b,\sigma)$ be as defined in~\eqref{eq:def-tildea} and~\eqref{eq:def-sigma}. By definition,
\begin{align*}
\qnorm{G\otimes \Id\ket{\tilde{A}}}^2&=\sup_{\ket{\Psi}}\, \Es{v\sim v'}\,\sum_{b\leftrightarrow b'} \,\bra{\Psi} \overline{\tilde{A}_v^b} \otimes \tilde{A}_{v'}^{b'}\ket{\Psi}\\
&\geq \Es{v\sim v'}\,\sum_{b\leftrightarrow b'} \,\Tr\big(\big( \overline{\tilde{A}_v^b} \otimes \tilde{A}_{v'}^{b'}\big)\, \sigma\big)\\
&= 1-O\big(\eta/\lambda),
\end{align*}
where the last line follows from~\eqref{eq:sigma-good}. Using $\qnorm{G}^2 \geq \qnorm{G\otimes \Id\ket{\tilde{A}}}^2$ concludes the proof of the lemma, with exponent $c=1$.
\end{proof}

It remains to prove Claim~\ref{claim:expand-bound}. The proof of the claim will use the expansion properties of $G$ through the following: 

\begin{claim}\label{claim:expand}
Suppose~\eqref{eq:aomega-val-expand} holds, and $G$ is such that the smallest nonzero eigenvalue of the Laplacian $L := \sum_v \ket{v}\bra{v} - \sum_{v,v'\sim v} \mu(v,v')\mu(v)^{-1/2}\mu(v')^{-1/2} \ket{v'}\bra{v}$ is at least $\lambda>0$. Then 
 \begin{equation}
 \Es{v, v'}\bra{\hat{\Psi}} \overline{\hat{A}_{\omega v}}\otimes \hat{A}_{\omega v'} \ket{\hat{\Psi}} \geq (1-2\eta/\lambda) \Es{v} \bra{\hat{\Psi}} \overline{\hat{A}_{\omega v}}\otimes \hat{A}_{\omega v} \ket{\hat{\Psi}}.\label{eq:assumption-2}
\end{equation}
\end{claim}

\begin{proof}
Using~\eqref{eq:ando-good} we can write  
$$\bra{\hat{\Psi}} \overline{\hat{A}_{\omega v}}\otimes \hat{A}_{\omega v'} \ket{\hat{\Psi}} \,=\, \Tr\big(\hat{A}_{\omega v} \rho^{1/2} \hat{A}_{\omega v'}\rho^{1/2} \big),$$
where $\rho$ is the reduced density of $\ket{\hat{\Psi}}$ on either subsystem. Let
$\tilde{L} := L\otimes \Id$ and $A := \sum_{v} \mu(v)^{1/2}\ket{v}\otimes \hat{A}_{\omega v} $. Using that~\eqref{eq:aomega-val-expand} holds for our choice of $\omega$,
\begin{align}
 \Tr\big(A^\dagger \tilde{L}(\Id\otimes \rho^{1/2}) A(\Id\otimes \rho^{1/2})\big) &= \Es{v}\, \Tr\big(\hat{A}_{\omega v} \rho^{1/2} \hat{A}_{\omega v} \rho^{1/2}\big) - \Es{v\sim v'} \Tr\big(\hat{A}_{\omega v} \rho^{1/2} \hat{A}_{\omega v'} \rho^{1/2}\big) \notag\\
&\leq \eta \Es{v} \Tr\big(\hat{A}_{\omega v} \rho^{1/2} \hat{A}_{\omega v} \rho^{1/2}\big).\label{eq:lapl1}
 \end{align}
The normalized Laplacian $L$ has smallest eigenvalue $0$, and second smallest $\lambda>0$. Let the smallest eigenvector of $L$ be $\ket{u_0} = \sum_v \mu(v)^{1/2} \ket{v}$, and write $A = \ket{u_0}\otimes A_0 + \sum_{i>0} \ket{u_i} \otimes A_i$, where the $\ket{u_i}$ are the remaining eigenvectors, with associated eigenvalue $\lambda_i$, of $\tilde{L}$, and $A_0 = \sum_v \mu(v)^{1/2} \hat{A}_{\omega v}$. Then 
\begin{align}
 A^\dagger \tilde{L} (\Id\otimes \rho^{1/2})A(\Id \otimes \rho^{1/2}) &= \sum_{i>0} \lambda_i A_i^\dagger \rho^{1/2} A_i \rho^{1/2}.\label{eq:lapl1b}
\end{align}
Taking the trace we get
\begin{align*}
\Es{v} \Tr\Big( \big(\hat{A}_{\omega v}- \Es{v'} \hat{A}_{\omega v'}\big) \rho^{1/2} \big(\hat{A}_{\omega v}- \Es{v'} \hat{A}_{\omega v'}\big)\rho^{1/2}\Big) &= \Tr\big((A - \ket{v_0}\otimes A_0)^\dagger \rho^{1/2} (A - \ket{v_0}\otimes A_0) \rho^{1/2}\big)\\
&= \sum_{i>0} \Tr\big(A_i^\dagger \rho^{1/2} A_i \rho^{1/2}\big)\\
&\leq \frac{\eta}{\lambda} \Es{v} \Tr\big(\hat{A}_{\omega v}\rho^{1/2} \hat{A}_{\omega v}\rho^{1/2}\big),
\end{align*}
where the last inequality follows from~\eqref{eq:lapl1b} and~\eqref{eq:lapl1}. 
\end{proof}

We conclude this section by giving the proof of Claim~\ref{claim:expand-bound}.

\begin{proof}[Proof of Claim~\ref{claim:expand-bound}]
For any four vertices $v,v',w$ and $w'$ define
$$ \eps_{vv'}^{ww'} \,:=\,  \sum_{b \nleftrightarrow b'} \bra{\Psi_{ww'}} \overline{\tilde{A}_v^b} \otimes \tilde{A}_{v'}^{b'} \ket{\Psi_{ww'}}.$$
Note also that, given our choice of the unitaries $U_v$ satisfying~\eqref{eq:defu} and using~\eqref{eq:ando-good},
\begin{align}
\bra{\Psi_{ww'}} \overline{\tilde{A}_v^b} \otimes \tilde{A}_{v'}^{b'} \ket{\Psi_{ww'}} &= \Tr\big(\rho^{1/2}(\hat{A}_{w'})^{1/2} \tilde{A}_v^b (\hat{A}_{w})^{1/2}\rho^{1/2} (\hat{A}_{w})^{1/2}\tilde{A}_{v'}^{b'} (\hat{A}_{w'})^{1/2}\big)\notag\\
&= \Tr\big( \hat{A}_{w'}\rho^{1/4} {\tilde{A}_v^b} \rho^{1/4} \hat{A}_{w}\rho^{1/4}\tilde{A}_{v'}^{b'} \rho^{1/4}\big),\label{eq:expand-uv}
\end{align}
an identity that will prove useful. 

By~\eqref{eq:def-eps} and the definition of $\sigma$ we have $\eps = (\Es{w}\|\ket{\Psi_{ww}}\|^{2})^{-1} \Es{v\sim v'} \Es{w} \eps_{vv'}^{ww}$. To prove the claim it will suffice to show that $\eps=O(\eta/\lambda \Es{w}\|\ket{\Psi_{ww}}\|^{2})$. 
Eq.~\eqref{eq:aomega-val-expand} implies that 
$$\Es{v\sim v'} \|\ket{\Psi_{vv'}}\|^2 \geq (1-\eta)\Es{v}\|\ket{\Psi_{vv}}\|^2,$$
hence (using~\eqref{eq:aomega-val-expand} once more) $\Es{v\sim v'} \eps_{vv'}^{vv'} =O(\eta\Es{w}\|\ket{\Psi_{ww}}\|^{2})$. We relate these quantities by establishing the following three bounds.
\begin{align}
 \Es{v\sim v'} \big| \eps_{vv'}^{vv'} - \eps_{vv'}^{vv}\big| &= O(\eta)\Es{v}\|\ket{\Psi_{vv}}\|^{2},\label{eq:epsvv-1}\\
 \Es{v\sim v'}\Es{w} \big| \eps_{vv'}^{vv'} - \eps_{vv'}^{wv'}\big| &= O\big(\eta/\lambda\big)\Es{v}\|\ket{\Psi_{vv}}\|^{2}\label{eq:epsvv-2},\\
 \Es{v\sim v'}\Es{w} \big| \eps_{vv'}^{wv'} - \eps_{vv'}^{ww}\big| &= O\big(\eta/\lambda\big)\Es{v}\|\ket{\Psi_{vv}}\|^{2}.\label{eq:epsvv-3}
\end{align}
It is clear that~\eqref{eq:epsvv-2} and~\eqref{eq:epsvv-3} together will conclude the proof. We first show~\eqref{eq:epsvv-1}. Using~\eqref{eq:expand-uv},
\begin{align*}
\Big|\eps_{vv'}^{vv'} - \eps_{vv}^{vv}\Big|
 &= \Big|\sum_{b \nleftrightarrow b'}  \Tr\big( (\hat{A}_{v'}-\hat{A}_{v})\rho^{1/4} {\tilde{A}_v^b} \rho^{1/4} \hat{A}_{v}\rho^{1/4}\tilde{A}_{v'}^{b'} \rho^{1/4}\big)\Big|\\
&\leq \Big(\Tr\big((\hat{A}_{v'}-\hat{A}_{v})\rho^{1/2}(\hat{A}_{v'}-\hat{A}_{v})\rho^{1/2}\big)\Big)^{1/2}\Big(\sum_{b \nleftrightarrow b'} \Tr\big(\hat{A}_{v}\rho^{1/4} {\tilde{A}_v^b} \rho^{1/4} \hat{A}_{v}\rho^{1/4}\tilde{A}_{v'}^{b'} \rho^{1/4}\big)\Big)^{1/2},
\end{align*}
where the inequality follows from applying the Cauchy-Schwarz inequality to 
$$ \big(\tilde{A}_{v'}^{b'} \big)^{1/2}\rho^{1/4}(\hat{A}_{v'}-\hat{A}_{v})\rho^{1/4} \big({\tilde{A}_v^b}\big)^{1/2}\qquad\text{and}\qquad \big({\tilde{A}_v^b}\big)^{1/2} \rho^{1/4}\hat{A}_{v}\rho^{1/4} \big(\tilde{A}_{v'}^{b'} \big)^{1/2}$$
and using $\sum_b \tilde{A}_{v}^{b}\leq \Id$. Taking the expectation over $v\sim v'$ and using~\eqref{eq:aomega-val-expand} together with $|x-y|\leq \sqrt{ax} \implies x\leq a+4y$ gives~\eqref{eq:epsvv-1}. To prove~\eqref{eq:epsvv-2}, write  
\begin{align*}
\Big|\eps_{vv'}^{vv'} - \eps_{vv'}^{wv'}\Big|
 &= \Big|\sum_{b \nleftrightarrow b'}  \Tr\big( \hat{A}_{v'}\rho^{1/4} {\tilde{A}_v^b} \rho^{1/4}(\hat{A}_{v}-\hat{A}_{w})\rho^{1/4}\tilde{A}_{v'}^{b'} \rho^{1/4}\big)\Big|\\
&\leq \Big(\Tr\big((\hat{A}_{v}-\hat{A}_{w})\rho^{1/2}(\hat{A}_{v}-\hat{A}_{w})\rho^{1/2}\big)\Big)^{1/2}\Big(\sum_{b \nleftrightarrow b'} \Tr\big(\hat{A}_{v'}\rho^{1/4} {\tilde{A}_v^b} \rho^{1/4} \hat{A}_{v'}\rho^{1/4}\tilde{A}_{v'}^{b'} \rho^{1/4}\big)\Big)^{1/2},
\end{align*}
where the inequality follows from a similar application of the Cahuchy-Schwarz inequality as performed above. The second term above is $\eps_{vv'}^{v'v'}$, so using~\eqref{eq:epsvv-1}, and~\eqref{eq:assumption-2} to bound the first term, we have proved~\eqref{eq:epsvv-2}. Finally, to prove~\eqref{eq:epsvv-3} write
\begin{align*}
\Big|\eps_{vv'}^{wv'} - \eps_{vv'}^{ww'}\Big|
 &= \Big|\sum_{b \nleftrightarrow b'}  \Tr\big( (\hat{A}_{v'}-\hat{A}_{w'})\rho^{1/4} {\tilde{A}_v^b} \rho^{1/4}\hat{A}_{v}\rho^{1/4}\tilde{A}_{v'}^{b'} \rho^{1/4}\big)\Big|\\
&\leq \Big(\Tr\big((\hat{A}_{v'}-\hat{A}_{w'})\rho^{1/2}(\hat{A}_{v'}-\hat{A}_{w'})\rho^{1/2}\big)\Big)^{1/2}\Big(\sum_{b \nleftrightarrow b'} \Tr\big(\hat{A}_{v}\rho^{1/4} {\tilde{A}_v^b} \rho^{1/4} \hat{A}_{v}\rho^{1/4}\tilde{A}_{v'}^{b'} \rho^{1/4}\big)\Big)^{1/2},
\end{align*}
which is again bounded using~\eqref{eq:epsvv-1} and~\eqref{eq:assumption-2}. Combining~\eqref{eq:epsvv-2} and~\eqref{eq:epsvv-3} proves the claim. 

\end{proof}


\subsection{Non-expanding games}\label{sec:non-expand}

Suppose $G$ is an arbitrary (not necessarily expanding) projection game. In the game $G^\dagger G$ the players Bob and Bob' are always sent neighboring $v\sim v'$. Using notation from the previous section, we would like to enable the players to take advantage of the possibility of using an arbitrary entangled state in order to initialize themselves in a  state that is close to $\ket{\Psi_{\omega vv'}}$. The difficulty is that this must be done ``on the fly'', as $\ket{\Psi_{\omega vv'}}$ depends on the questions $v,v'$; indeed since $G$ is not expanding there may not be a single state close to all $\ket{\Psi_{\omega vv'}}$ that they could have agreed upon before the start of the game; for instance $G$ could be a direct sum of two independent games for which the optimal entangled state and measurements need not bear any relation to each other.

To get around this we resort to the use of a so-called family of ``universal embezzling states'' $\ket{\Gamma_d}\in\C^d\otimes \C^d$. These states, introduced in~\cite{DamH03embezzlement}, have the property that for any given state $\ket{\psi}$ there exists a $d$ and unitaries $U,V$ such that $U\otimes V \ket{\Gamma_d} \approx \ket{\psi} \ket{\Gamma_{d'}}$ for some $d'$. Hence if \emph{both} players have a description of the target state $\ket{\Psi_{\omega vv'}}$ they can easily generate it locally from the universal state  $\ket{\Gamma_d}$. Our setting presents and additional difficulty: only the first player, Bob, knows $v$, and the second, Bob', knows $v'$; how to make them agree on which state to embezzle? We will use the following lemma. 

\begin{lemma}\label{lem:psi-close}
Let $\ket{\Phi}$ be a bipartite state invariant under permutation of the two subsystems, $\rho$ its reduced density on either subsystem, $0\leq A_v \leq \Id$, and $\nu$ a distribution on $\mathcal{V}\times\mathcal{V}$ that is symmetric under permutation of the two coordinates (we also denote by $\nu$ the marginal distribution on either coordinate), such that
\beq\label{eq:close-assumption}
 \Es{(v, v')\sim\nu} \bra{\Phi} \overline{A_v} \otimes A_{v'} \ket{\Phi} \geq (1-\eta) \Es{v\sim\nu}\bra{\Phi} \overline{A_v} \otimes A_{v} \ket{\Phi}.
\eeq
Let $U_v$ be unitaries such that 
\begin{equation}\label{eq:defu2}
U_{v}A_v^{1/2} \rho^{1/4} =  \rho^{1/4}A_{v}^{1/2} U_{v}^\dagger = \big(\rho^{1/4}A_{v}\rho^{1/4} \big)^{1/2},
\end{equation}
and let
$$ \ket{\Phi_{v v'}} := \overline{U_v}\ \overline{A_v}^{1/2} \otimes {U_{v'}} A_{v'}^{1/2} \ket{\Phi}.$$
Then for any $v''\in \mathcal{V}$,
$$ \Es{v\sim v'}\, \big\|\ket{{\Phi}_{v v''}} -\ket{{\Phi}_{ v'v''}} \big\|^2 \,=\, O(\eta^{1/2})\Big(\Es{v}\big\|\ket{{\Phi}_{ vv}} \big\|^2\Big)^{1/2}\big\|\ket{{\Phi}_{ v''v''}} \big\|,$$
and
$$ \Es{v\sim v'}\, \big\|\ket{{\Phi}_{v v}} -\ket{{\Phi}_{ v'v}} \big\|^2 \,=\, O(\eta^{1/2})\,\Es{v}\,\big\|\ket{{\Phi}_{ vv}} \big\|^2.$$
\end{lemma}

The lemma is stated in a stand-alone form, but we may apply it to the present setting by letting $\ket{\Phi}$ be the state $\ket{\hat{\Psi}}$ and $A_v$ the measurement operators $\hat{A}_{\omega v}$ (for some $\omega$) whose existence is guaranteed by Claim~\ref{claim:aomega}. Recalling the definition of $\ket{\Psi_{\omega vv'}}$ in~\eqref{eq:def-tildepsi}, Lemma~\ref{lem:psi-close} (together with~\eqref{eq:aomega-val} to obtain~\eqref{eq:close-assumption}) implies that (for most $\omega$)
$$ \Es{v\sim v'} \,\big\|\ket{\Psi_{\omega vv'}} - \ket{\Psi_{\omega vv}} \big\|^2 \,=\, O(\eta^{1/2}) \Es{v} \,\big\| \ket{\Psi_{\omega vv}}\big\|^2,$$
that is, all three states $\ket{\Psi_{\omega vv'}}$, $\ket{\Psi_{\omega vv}}$ and $\ket{\Psi_{\omega v'v'}}$ are close for neighboring $v\sim v'$. Hence the first player, knowing his  question $v$, can compute a classical description of the state $\ket{{\Psi}_{\omega v v}}$; the second player can compute a classical description of $\ket{{\Psi}_{\omega v' v'}}$. These two states are close to each other as well as to the target state: are these conditions sufficient for the two players to successfully embezzle a joint state close to either of the three?

It turns out that, if one na\"ively applies the embezzling procedure described in~\cite{DamH03embezzlement}, it can fail completely even when the states are arbitrarily close (see Section~\ref{sec:correlated} for an example). Nevertheless, in the next section we state and prove a ``quantum correlated sampling lemma'', which extends the results in~\cite{DamH03embezzlement} to this ``approximate'' scenario.

We first prove Lemma~\ref{lem:psi-close}, and then show how the lemma, together with the correlated sampling lemma, Lemma~\ref{lem:correlated}, imply Lemma~\ref{lem:approx} for the case of general games.

\begin{proof}[Proof of Lemma~\ref{lem:psi-close}]
Let $X_v$ be defined as
\begin{equation}\label{eq:xu-def}
X_v \,:=\, U_{v}A_{v}^{1/2} \rho^{1/4}\,=\,\rho^{1/4}A_{v}^{1/2}U_{v}^\dagger.
\end{equation}
Using~\eqref{eq:defu2}, $X_v$ is positive semidefinite. With this notation we have the following useful identities.

\begin{claim}\label{claim:xu}
For every $v,v'\in\mathcal{V}$ we have
\begin{equation}\label{eq:xu-1}
\Tr(X_v^4) = \Tr((X_vX_v^\dagger)^2) \,=\, \bra{\Phi} \overline{A_v} \otimes A_{v} \ket{\Phi} = \|\ket{\Phi_{ vv}}\|^2
\end{equation}
and
\begin{equation}\label{eq:xu-2}
\Tr\big( X_v^2  X_{v'}^2 \big) \,=\, \bra{\Phi}\overline{A_{v}} \otimes A_{v'}\ket{\Phi}.
\end{equation}
\end{claim}

\begin{proof}
For~\eqref{eq:xu-1} we use the definition of $X_v$ to write
$$\Tr(X_v^4) = \Tr((X_vX_v^\dagger)^2) \,=\, \Tr(A_{v} \rho^{1/2}A_{v} \rho^{1/2}) = \bra{\Phi} \overline{A_v} \otimes A_{v} \ket{\Phi},$$
where the last equality follows from Ando's identity, Claim~\ref{claim:ando}, together with our assumption on $\ket{\Phi}$ being permutation-invariant.
To show~\eqref{eq:xu-2}, expand using the definition~\eqref{eq:xu-def}
\begin{align*}
\Tr(X_v^2  X_{v'}^2) &= \Tr(U_{v}A_v^{1/2} \rho^{1/2}A_{v}^{1/2}U_{v}^\dagger U_{v'} A_{v'}^{1/2} \rho^{1/2}A_{v'}^{1/2}U_{v'}^\dagger) \\
&=  \Tr\big( A_{v} \rho^{1/2} A_{v'}\rho^{1/2}\big)\\
&= \bra{\Phi}\overline{A_{v}} \otimes A_{v'} \ket{\Phi},
\end{align*}
where the second equality follows from~\eqref{eq:defu2} and the last from Claim~\ref{claim:ando}.
\end{proof}

Now for any three $v,v',v''$,
\begin{align}
 \| \ket{\Phi_{ vv''}} - \ket{\Phi_{ v'v''}} \|^2 &=\big(\bra{\Phi_{ vv''}} - \bra{\Phi_{ v'v''}}\big)\big( \ket{\Phi_{ vv''}} - \ket{\Phi_{ v'v''}}\big)\notag\\
&= \bra{\Phi} \big(\overline{A_{v}^{1/2}U_{v}^\dagger} - \overline{A_{v'}^{1/2}U_{v'}^\dagger}\big)\big(\overline{U_{v}A_{v}^{1/2}} -\overline{U_{v'} A_{v'}^{1/2}}\big) \otimes A_{v''}^{1/2} U_{v''}^\dagger U_{v''} A_{v''}^{1/2}\ket{\Phi}\notag\\
&= \Tr\big( (X_v - X_{v'})^\dagger (X_v-X_{v'}) X_{v''}^\dagger X_{v''} \big)\notag\\
&\leq \big(\Tr\big( (X_v - X_{v'})^4 \big) \big)^{1/2}\big(\Tr\big( X_{v''}^4\big)\big)^{1/2},\label{eq:xu-4}
\end{align}
where the last inequality follows from Cauchy-Schwarz and the fact that the $X_v$ are positive semidefinite. The first term on the right-hand side of~\eqref{eq:xu-4} can be bounded as
\begin{align*}
\Tr\big( (X_v - X_{v'})^4 \big) &\leq  \Tr\big( (X_v^2 - X_{v'}^2)^2\big)\\
&= \bra{\Phi}\overline{A_{v}} \otimes A_{v} \ket{\Phi} + \bra{\Phi}\overline{A_{v'}} \otimes A_{v'} \ket{\Phi} - 2\bra{\Phi}\overline{A_{v}} \otimes A_{v'}\ket{\Phi},
\end{align*}
where the first inequality can be found as e.g. Corollary~2 in~\cite{Kittaneh86lp} and the equality follows from~\eqref{eq:xu-1} and~\eqref{eq:xu-2}. Going back to~\eqref{eq:xu-4}, we obtain
\begin{align}
 \Es{v\sim v'}  \| \ket{\Phi_{ vv''}} - \ket{\Phi_{ v'v''}} \|^2 &\leq  \Big(2\eta \Es{v}\|\ket{\Phi_{vv}}\|^2 \Big)^{1/2}\,\Big( \|\ket{\Phi_{v''v''}}\|^2\Big)^{1/2},\notag
\end{align}
where the first inequality uses the assumption made in the lemma to bound the first term in~\eqref{eq:xu-4} and~\eqref{eq:xu-1} to rewrite the second. This proves the first inequality claimed in the lemma. The second is obtained by taking $v''=v$ in~\eqref{eq:xu-4}, and then the expectation over $v\sim v'$ as in the above.
\end{proof}

We conclude this section with the proof of Lemma~\ref{lem:approx}.

\begin{proof}[Proof of Lemma~\ref{lem:approx}, general case]
Let $\ket{\hat{A}}$ be a vector strategy, and $\ket{\hat{\Psi}}$ a state such that~\eqref{eq:aomega-val} holds. Our goal is to identify a quantum strategy $\ket{\tilde{A}}$ such that $\qnorm{G\otimes \Id \ket{\tilde{A}}}^2 \geq 1-O(\eta^{1/c})$, which by Claim~\ref{claim:aomega} will suffice to prove Lemma~\ref{lem:approx}.

We define a ``re-normalized'' vector strategy $\ket{\tilde{A}}\in \C^{|\Omega|}\otimes \C^{|\mathcal{V}|}\otimes \C^{|\mathcal{B}|}\otimes \lin(\C^d)$, from which we will later obtain a quantum strategy $\ket{\tilde{A}_\omega}$ by making a good choice of $\omega\in \Omega$. As previously, for every $\omega$ we may define states
\beq\label{eq:def-tildepsi-omega}
\ket{\Psi_{\omega vv'}}\,:=\, \overline{U_{\omega v} }\overline{\hat{A}_{\omega v}}^{1/2} \otimes U_{\omega v'}\hat{A}_{\omega v'}^{1/2} \ket{\hat{\Psi}},
\eeq
where the $U_{\omega v}$ are the unitaries given by Lemma~\ref{lem:psi-close}: as a consequence of~\eqref{eq:aomega-val} (replacing the max on the right-hand-side by an average) the assumption of the lemma is satisfied, on average over $\omega\in\Omega$, for the states $\ket{\Psi_{\omega v v'}}$. 
The lemma gives the following bound:
\begin{equation}\label{eq:psi-nonexp}
\Es{\omega}\Es{v\sim v'} \|\ket{{\Psi}_{\omega v v}} -\ket{{\Psi}_{ \omega v v'}} \|^2 \,=\, O(\eta^{1/2})\Es{\omega}\Es{v}\|\ket{\Psi_{\omega vv}}\|^2.
\end{equation}
 In addition, for every $\omega$ and question $v\in\mathcal{V}$ let $\overline{V_{\omega v}}$ and $W_{\omega v}$ be the unitaries that are defined in Lemma~\ref{lem:correlated}, for the (re-normalized) state $\ket{\Psi_{\omega vv}}$ and a choice of $\delta = \eta^2$. By convexity the lemma gives us that
\begin{equation}\label{eq:single-general-0}
\Es{\omega}\Es{v\sim v'} \|\overline{V_{\omega v}}\otimes W_{{\omega v'}} \ket{\Gamma_{dd'}} - \|\ket{\Psi_{\omega vv}}\|^{-1}\ket{\Psi_{\omega vv}}\ket{\Gamma_{d'}}\|^2 \,=\, O\Big(\Es{\substack{\omega\\v\sim v'}}\Big\|\frac{\ket{\Psi_{\omega vv}}}{\|\ket{\Psi_{\omega vv}}\|} - \frac{\ket{\Psi_{\omega v'v'}}}{\|\ket{\Psi_{\omega v'v'}}\|}\Big\|^{2/6}\Big).
\end{equation}
For any question $v\in\mathcal{V}$ and answer $b\in\mathcal{B}$, define measurement operators
$$\tilde{A}_{\omega v}^b:= V_{\omega v}^\dagger \big(U_{\omega v} \hat{A}_{\omega v}^{-1/2} \hat{A}_{\omega v}^b \hat{A}_{\omega v}^{-1/2} U_{\omega v}^\dagger\otimes \Id_{d'}\big) V_{\omega v},\, \tilde{B}_{\omega v}^b:= W_{\omega v}^\dagger \big(U_{\omega v} \hat{A}_{\omega v}^{-1/2} \hat{A}_{\omega v}^b \hat{A}_{\omega v}^{-1/2} U_{\omega v}^\dagger\otimes \Id_{d'}\big) W_{\omega v}.$$
It is easy to verify that each $\tilde{A}_{\omega v}^b$ and $\tilde{B}_{\omega v}^b$ is positive semidefinite, and that $\sum_b \tilde{A}_{\omega v}^b,\,\sum_b \tilde{B}_{\omega v}^b \leq \Id$. Since we may always add a ``dummy'' outcome in order for the measurement operators to sum to identity, both $\{\tilde{A}_{\omega v}^b\}_b$ and $\{\tilde{B}_{\omega v}^b\}_b$ are easily made into well-defined measurements, and for every $\omega$, $\ket{\tilde{A}_\omega} := \sum_{v,b} \ket{v,b}\otimes \tilde{A}_{\omega v}^b$ and $\ket{\tilde{B}_\omega} := \sum_{v,b} \ket{v,b}\otimes \tilde{B}_{\omega v}^b$ valid strategies for the  players Bob and Bob' in $G^\dagger G$ (we will soon show that at least one of these strategies must be a good strategy for the square game).

We first bound
\begin{align}
\Es{\omega}&\Es{v\sim v'} \|\overline{V_{\omega v}}\otimes W_{{ \omega v'}} \ket{\Gamma_{dd'}} - \|\ket{\Psi_{\omega vv}}\|^{-1}\ket{\Psi_{\omega vv'}}\ket{\Gamma_{d'}}\|^2\notag\\
&\leq \Es{\omega}\Es{v\sim v'} \|\overline{V_{\omega v}}\otimes W_{{\omega v'}} \ket{\Gamma_{dd'}} - \|\ket{\Psi_{\omega vv}}\|^{-1}\ket{\Psi_{\omega vv}}\ket{\Gamma_{d'}}\|^2
 + \Es{\omega}\Es{v\sim v'} \|\ket{\Psi_{\omega vv}}\|^{-2}\|\ket{\Psi_{\omega vv}} - \ket{\Psi_{\omega vv'}}\|^2 \notag\\
&= O\Big(\Es{\omega}\Es{v\sim v'}\,  \Big\|\frac{\ket{\Psi_{\omega vv}}}{\|\ket{\Psi_{\omega vv}}\|} - \frac{\ket{\Psi_{\omega v'v'}}}{\|\ket{\Psi_{\omega v'v'}}\|}\Big\|^{2/6}\Big) + O(\eta^{1/2})\notag\\
&=  O\big(\Es{\omega}\Es{v\sim v'} \|\ket{\Psi_{\omega vv}}\|^{-1/6} \|\ket{\Psi_{\omega v'v'}}\|^{-1/6}\|\ket{\Psi_{\omega vv}} - \ket{\Psi_{\omega v'v'}}\|^{2/6}\big) + O(\eta^{1/2})\notag\\
&= O(\eta^{1/{6}}),\label{eq:single-general-1}
\end{align}
where in the second line we used~\eqref{eq:psi-nonexp} and the Cauchy-Schwarz inequality to bound the last term, and~\eqref{eq:single-general-0} for the first; in the third line we used that $\|\ket{\Psi_{vv}}\|\leq 1$,
and in the last we again applied~\eqref{eq:psi-nonexp} and the Cauchy-Schwarz inequality. Note that
\begin{align}
\Es{\omega}\,\qnorm{ G \otimes \Id \ket{\tilde{A}_\omega} }\qnorm{ G \otimes \Id \ket{\tilde{B}_\omega}} &= \Big\| \Es{\omega}\Es{v\sim v'} \sum_{b\leftrightarrow b'}  \overline{\tilde{A}_{\omega v}^b} \otimes \tilde{A}_{\omega v'}^{b'} \Big\|_\infty^{1/2}\Big\| \Es{\omega}\Es{v\sim v'} \sum_{b\leftrightarrow b'}  \overline{\tilde{B}_{\omega v}^b} \otimes \tilde{B}_{\omega v'}^{b'} \Big\|_\infty^{1/2} \notag\\
&\geq \Big\| \Es{\omega}\Es{v\sim v'} \sum_{b\leftrightarrow b'}  \overline{\tilde{A}_{\omega v}^b} \otimes \tilde{B}_{\omega v'}^{b'} \Big\|_\infty,\notag
\end{align}
where the last inequality follows from Claim~\ref{claim:haagerup}. Hence
\begin{align}
\qnorm{G}^2 &\geq \Es{\omega}\,\qnorm{ G \otimes \Id \ket{\tilde{A}_\omega} }\qnorm{ G \otimes \Id \ket{\tilde{B}_\omega} }\notag\\
&\geq \Es{\omega}\Es{v\sim v'} \sum_{b\leftrightarrow b'} \bra{\Gamma_{dd'}} \overline{\tilde{A}_{\omega v}^b} \otimes \tilde{B}_{\omega v'}^{b'} \ket{\Gamma_{dd'}}\notag\\
&\geq \Es{\omega} \Es{v\sim v'} \sum_{b\leftrightarrow b'}\|\ket{\Psi_{\omega vv}}\|^{-2}\bra{\Psi_{\omega vv'}}\overline{U_{\omega v}^\dagger\hat{A}_{\omega v}^{-1/2} \hat{A}_{ \omega v}^b \hat{A}_{\omega v}^{-1/2}U_{\omega v}} \notag\\
&\hskip5cm\otimes U_{\omega v'}\hat{A}_{\omega v'}^{-1/2} \hat{A}_{ \omega v'}^{b'} \hat{A}_{\omega v'}^{-1/2}U_{\omega v'}^\dagger \ket{\Psi_{\omega vv'}}-O(\eta^{1/12})\notag\\
&= \Es{\omega} \Es{v\sim v'} \sum_{b\leftrightarrow b'}\|\ket{\Psi_{\omega vv}}\|^{-2}\bra{\hat{\Psi}}\overline{\hat{A}_{ \omega v}^b} \otimes \hat{A}_{\omega  v'}^{b'} \ket{\hat{\Psi}}-O(\eta^{1/12}),\label{eq:single-general-2}
\end{align}
where the second line uses the definition of $\tilde{A}_{\omega v}^b$ and~\eqref{eq:single-general-1} and the third is by definition of $\ket{\Psi_{\omega vv'}}$. To conclude, note that applying Markov's inequality to~\eqref{eq:aomega-val} we get that a fraction at least $1-\eta^{1/3}$ of $v\sim v'$ are such that
$$ \Es{\omega} \sum_{b\leftrightarrow b'}\,\bra{\hat{\Psi}}\overline{\hat{A}_{ \omega v}^b} \otimes \hat{A}_{\omega  v'}^{b'} \ket{\hat{\Psi}} \geq (1-\eta^{2/3})\Es{\omega}\|\ket{\Psi_{\omega vv}}\|^{2},$$
where here we crucially used the $\max$ on the right-hand side of~\eqref{eq:aomega-val} to allow ourselves use the same $v$ on the right-hand side as on the left-hand side.
For any such $v\sim v'$, a fraction $1-\eta^{1/3}$ of $\omega\in\Omega$ will be such that
$$\sum_{b\leftrightarrow b'}\,\bra{\hat{\Psi}}\overline{\hat{A}_{ \omega v}^b} \otimes \hat{A}_{\omega  v'}^{b'} \ket{\hat{\Psi}} \geq (1-\eta^{1/3})\|\ket{\Psi_{\omega vv}}\|^{2}.$$
For these $v\sim v'$ and $\omega$ the right-hand side of~\eqref{eq:single-general-2} is at least $1-\eta^{1/3}-O(\eta^{1/12})$, and their total weight constitutes at least an $(1-2\eta^{1/3})$ fraction of the total.
\end{proof}

\section{The correlated sampling lemma}\label{sec:correlated}

In this section we prove our quantum correlated sampling lemma.

\begin{lemma}\label{lem:correlated}
Let $d$ be an integer and $\delta>0$. There exists an integer $d'$, and for every state $\ket{\psi}\in \C^d\otimes \C^d$ unitaries $V_\psi$, $W_\psi$ acting on $\C^{dd'}$, such that the following holds for any two states $\ket{\psi},\ket{\varphi}\in \C^d\otimes \C^d$:
$$\| \overline{V_\psi} \otimes W_\varphi \ket{\Gamma_{dd'}} - \ket{\psi}\ket{\Gamma_{d'}}\| \, = \, O\big(\max\big\{\delta^{1/12},\|\ket{\psi}-\ket{\varphi}\|^{1/6}\big\}\big),\footnote{Note that here we implicitly re-ordered the registers, and $\ket{\psi}\ket{\Gamma_{d'}}$ should be understood as a bipartite state in $\C^{dd'}\otimes \C^{dd'}$, with the first (resp. second) space $\C^{dd'}$ being associated with the tensor product of the first (resp. second) spaces, $\C^d$ and $\C^{d'}$, respectively associated with $\ket{\psi}$ and $\ket{\Gamma_{d'}}$.}$$
where here $\ket{\Gamma_d} \propto \sum_{1\leq i \leq d} i^{-1/2} \ket{i}\ket{i}$ is the (properly normalized) $d$-dimensional embezzlement state.
\end{lemma}

A variant of the lemma holding for the special case of $\ket{\psi}=\ket{\varphi}$ was shown in~\cite{DamH03embezzlement}, where the ``embezzlement state'' $\ket{\Gamma_d}$ was first introduced. It is not hard to see however that the construction of the unitaries $V_\psi$, $W_\varphi$ given in that paper does not satisfy the conclusion of Lemma~\ref{lem:correlated}. For instance, if $\ket{\psi}=\sqrt{(1+\eps)/2}\ket{00}+\sqrt{(1-\eps)/2}\ket{11}$ and  $\ket{\varphi}=\sqrt{(1-\eps)/2}\ket{00}+\sqrt{(1+\eps)/2}\ket{11}$) then one can check that for any $\eps>0$ the unitaries from~\cite{DamH03embezzlement} will be such that $\| \overline{V_\psi} \otimes W_\varphi  \ket{\Gamma_{2d'}} - \ket{\psi}\ket{\Gamma_{d'}}\|\geq 1/4$. This is due to our taking advantage of the degenerate spectrum of the reduced density of the EPR pair $(\ket{00}+\ket{11})/\sqrt{2}$ to split the spectrum of the reduced density matrices of the nearby states $\ket{\psi},\ket{\varphi}$ in two different ways; our proof of Lemma~\ref{lem:correlated} shows that this is essentially the only obstacle that needs to be overcome in order to obtain a robust correlated sampling procedure.

Lemma~\ref{lem:correlated} can be seen as a quantum analogue of Holenstein's correlated sampling lemma~\cite{Hol09}, which played an important role in his proof of the classical parallel repetition theorem. There the players, Alice and Bob, receive as inputs a description of a distribution $p$, $q$ respectively such that $\|p-q\|_1 = \delta$. Their goal is to sample an element $u\sim p$ for Alice, $v\sim q$ for Bob, such that $u=v$ with probability $1-O(\delta)$. This task can be reproduced in our setting by giving the states $\ket{\psi}= \sum_u \sqrt{p(u)}\ket{u}\ket{u}$ to Alice and $\ket{\phi}= \sum_v \sqrt{q(v)}\ket{v}\ket{v}$ to Bob. If the players run our procedure and then measure their joint state in the computational basis they will obtain samples with a distribution close to $p$ and $q$, and moreover these samples will be identical with high probability (though our proof would require them to use entanglement in order to do so!).

After the completion of this work Anshu et al.~\cite{Jain14correlated} proposed a different quantum generalization of the classical correlated sampling lemma. In the task they consider the players are given reduced density matrices $\sigma$, $\tau$ respectively such that $\|\sigma-\tau\|_1 = \delta$. Their task is to generate a shared state $\ket{\Psi}_{AA'BB'}$, where Alice holds registers $AA'$ and Bob registers $BB'$, such that the reduced density of $\ket{\Psi}$ on $A$ (resp. $B$) is $\sigma$ (resp. $\tau$), and furthermore $\ket{\Psi}$ is close to being maximally entangled between $AA'$ and $BB'$. This task does not seem directly related to the one we consider; in particular, Anshu et al. show how their task can be accomplished starting from a sufficiently large number of shared EPR pairs while our task provably requires a universal embezzlement state to be successfully accomplished.\footnote{Indeed, local operations alone cannot change the Schmidt coefficients, and local operations on a maximally entangled state will only yield maximally entangled states (possibly of varying dimension). See~\cite{Leung13emb} for further discussion of the criteria for universal embezzlement.}

We note that we have not tried to optimize the parameters appearing in the lemma. In particular, from our proof one can verify that taking $d'=2^{O((d/\delta)^2)}$ in the lemma is sufficient, but this is probably far from optimal. Indeed, the method in~\cite{DamH03embezzlement} gives $d'=d^{O(1/\delta)}$; it may be possible to achieve such a polynomial dependence on $d$ here as well. (We refer the interested reader to recent work by Leung and Wang~\cite{Leung13emb} for an investigation of optimal families of embezzlement states, in the sense of van Dam and Hayden.)

\begin{proof}[Proof of Lemma~\ref{lem:correlated}]
We define the unitaries  $\overline{V_\psi}$, $W_\varphi$ implicitly through the following procedure, in which two players Alice, Bob receive classical descriptions of two bipartite states $\ket{\psi}$, $\ket{\varphi}$ respectively, each of local dimension $d$, as well as a precision parameter $\delta>0$. The unitaries $\overline{V_\psi}$ and $W_\varphi$ correspond to their respective local quantum operations as described in the procedure. The players' initial state consists of a classical description of the states $\ket{\psi}$, $\ket{\varphi}$ respectively (where each coefficient is specified with $\poly\log(\delta,d^{-1})$ bits of precision),  a large supply of private qubits initialized in the $\ket{0}$ state, a large supply of shared EPR pairs that they will use as classical shared randomness, and an embezzlement state $\ket{\Gamma_{dd'}}$ for some large enough $d'$.
\begin{enumerate}
\item Let $d$ be the local dimension of $\ket{\psi}$ and $\ket{\varphi}$, $\delta$ the precision parameter given as part of the input, and $\eta>0$ a small parameter to be specified later.
\item Using shared randomness, the players jointly compute a sequence $\tau_0,\ldots, \tau_{K+1}$, where $K=\big\lceil\frac{\log(d/\delta)}{\log(1+\eta)}\big\rceil$, as follows. They set $\tau_0 = 1$, $\tau_{K+1}=0$, and for $k=1,\ldots,K$ they jointly sample $\tau_k$ uniformly at random in the interval $[(1+\eta)^{-k},(1+\eta)^{-k+1})$.
\item Both players individually compute a classical description of the same (normalized) state
$$\ket{\xi_0}\,\propto\,\sum_{k=0}^K \tau_k \ket{k,k}_{AB} \ket{\Phi_d}_{AB},$$
 where $\ket{\Phi_d}=\sum_{i=1}^d \ket{i}\ket{i}$ is the un-normalized maximally entangled state on $\C^d\otimes \C^d$. Let $N = \big\lceil(2\delta d\sum_k \tau_k^2)^{-2}\big\rceil$. Alice and Bob jointly generate  $N$ copies of $\ket{\xi_0}$, which they can achieve using the universal embezzling procedure from~\cite{DamH03embezzlement} provided $d'$ is large enough.
\item Alice (resp. Bob) computes the Schmidt decomposition $\ket{\psi} = \sum_i \lambda_i \ket{u_i}\ket{u'_i}$ (resp. $\ket{\varphi} = \sum_i \mu_i \ket{v_i}\ket{v'_i}$). She sets $S_k$ (resp. $T_k$) as the set of those indices $i$ such that $\lambda_i\in [\tau_{k+1},\tau_{k})$ (resp. $\mu_i\in [\tau_{k+1},\tau_{k})$), $s_k=|S_k|$ (resp. $t_k=|T_k|$), and $P_k$ (resp. $Q_k$) the projector on the the span of the $\ket{u_i}$ for $i\in S_k$ (resp. $\ket{v_i}$ for $i\in T_k$).
\item Alice measures her share of the first copy of $\ket{\xi_0}$ using the two-outcome measurement $\{P_A,\Id-P_A\}$ where $P_A := \sum_k \ket{k}\bra{k}\otimes P_k$. Bob proceeds similarly with $P_B:=\sum_k \ket{k}\bra{k}\otimes Q_k$. If either of them obtains the first outcome they proceed to the next step. Otherwise, they repeat this step with the next copy of $\ket{\xi_0}$. If either player has used up all his or her copies he or she aborts the protocol.
\item Alice (resp. Bob) controls on the second register of $\ket{\xi_0}$ to erase $\ket{k}$ in the first register. (This is possible since the $P_k$ (resp. $Q_k$) are orthogonal projections.) The players discard all qubits but the remaining register of $\ket{\xi_0}$. Bob applies the unitary map $\ket{v_i}\to\ket{v'_i}$ to his share.
\end{enumerate}

Throughout the analysis we assume without loss of generality that $\delta \geq \|\ket{\psi}-\ket{\varphi}\|^2$. We will show that with probability at least $1-O(\delta^{1/12})$ the procedure described above results in a shared state between Alice and Bob that is within trace distance $O(\delta^{1/12})$ of both $\ket{\psi}$ and $\ket{\varphi}$. Our first claim shows that, based on the $\tau_k$, the players can each compute a discretized version of their inputs that both have (a slightly re-scaled version of) the $\tau_k$ as Schmidt coefficients.

\begin{claim}\label{claim:rounded-dist}
Define
$$ \ket{\Psi} := C\sum_k \tau_k \sum_{i\in S_k} \ket{u_i}\ket{u'_i} \qquad\text{and}\qquad \ket{\Phi} := C'\sum_k \tau_k \sum_{i\in T_k} \ket{v_i}\ket{v'_i},$$
where the $\tau_k$, $S_k$ and $T_k$ are as defined in the protocol and $C,C'$ are appropriate normalization constants. Then
\begin{equation}\label{eq:corr-3}
(1+\eta)^{-1}\,\leq\,C,C' \,\leq\, 1,
\end{equation}
and
\beq\label{eq:corr-00}
 \max\big\{\|\ket{\psi}-\ket{\Psi}\|^2,\,\|\ket{\varphi}-\ket{\Phi}\|^2\big\} \,=\,O(\eta).
\eeq
\end{claim}

\begin{proof}
We have $C^{-2} = \sum_k \tau_k^2 s_k$ which by definition of $S_k$ satisfies
$$ 1=\sum_i \lambda_i^2 \,\leq\,\sum_k \tau_k^2 s_k \, \leq \,\sum_i (1+\eta)^2\lambda_i^2 \,\leq \,(1+\eta)^2. $$
A similar calculation holds for $C'$, proving~\eqref{eq:corr-3}.
Next we bound the first term in~\eqref{eq:corr-00}, the second being similar. Using the definition of $\ket{\Psi}$ and~\eqref{eq:corr-3} we have
\begin{align*}
\|\ket{\psi}-\ket{\Psi}\|^2 &\leq \sum_k \sum_{i\in S_k} (\lambda_i - \tau_k)^2 + O(\eta)\\
&\leq \sum_k \sum_{i\in S_k} \tau_k^2\, \Big(1-\frac{1}{1+\eta}\Big)^2 + O(\eta)\\
&= O(\eta).
\end{align*}
\end{proof}

Our next claim shows that the subspaces $P_k,Q_k$ computed by the players are close, in the following sense.

\begin{claim}\label{claim:pq-close}
The following holds with probability at least $1-O(\delta^{1/6}\eta^{-1/3})$ over the choice of the $\tau_k$:
\beq\label{eq:pq-close-00}
 \sum_k \,\tau_k^2\, \Tr(P_kQ_k) \,=\, 1-O\big(\delta^{1/6}\eta^{-1/3}\big).
\eeq
\end{claim}

\begin{proof}
Using Claim~\ref{claim:rounded-dist} and $\|\ket{\psi}-\ket{\varphi}\|^2\leq \delta$ we deduce that $|\bra{\Phi}\Psi\rangle|^2 = CC'\sum_{k,k'} \tau_k\tau_k' \Tr(P_kQ_{k'}) = 1-O(\eta)$. To prove the claim we bound the contribution of those terms for which $k\neq k'$:
\begin{align}
\sum_{k\neq k'} \tau_k\tau_{k'} \Tr(P_kQ_{k'}) &= \sum_{k\neq k'} \tau_k\tau_{k'} \sum_{i\in S_k}\sum_{j\in T_{k'}} |\bra{u_i}v_j\rangle|^2\notag\\
&\leq (1+\eta)^2 \Big(\sum_{\substack{k\neq k',\,i\in S_k,j\in T_{k'}\\ |\sqrt{\lambda_i/\mu_j}-\sqrt{\mu_j/\lambda_i}|^2\geq \theta}} \lambda_i\mu_j |\bra{u_i}v_j\rangle|^2 +\sum_{\substack{k\neq k',z,i\in S_k,j\in T_{k'}\\ |\sqrt{\lambda_i/\mu_j}-\sqrt{\mu_j/\lambda_i}|^2< \theta}} \lambda_i\mu_j |\bra{u_i}v_j\rangle|^2\Big),\label{eq:pq-close-1}
\end{align}
where $\theta>0$ is a parameter to be fixed later. We bound each of the two terms inside the brackets in~\eqref{eq:pq-close-1} separately. The first term is at most
\begin{align*}
\sum_{\substack{i,j\\ |\sqrt{\lambda_i/\mu_j}-\sqrt{\mu_j/\lambda_i}|^2\geq \theta}} \lambda_i\mu_j |\bra{u_i}v_j\rangle|^2 &\leq  \sum_{\substack{i,j\\ |\lambda_i-\mu_j|^2\geq \theta\lambda_i\mu_j}} \frac{|\lambda_i-\mu_j|^2 }{\theta}|\bra{u_i}v_j\rangle|^2\\
&\leq \theta^{-1} \sum_{i,j} |\lambda_i-\mu_j|^2 |\bra{u_i}v_j\rangle|^2\\
&\leq \theta^{-1}\|\psi-\varphi\|^2\\
&\leq \delta\theta^{-1}.
\end{align*}
To bound the second term in~\eqref{eq:pq-close-1}, note first that provided $\theta$ is at most a small constant times $\eta$ necessarily $k'=k+1$ or $k'=k-1$; our choice of $\theta$ will satisfy this condition. Suppose $k'=k-1$, the other case being similar. Fix $i,j$ such that $|\sqrt{\lambda_i/\mu_j}-\sqrt{\mu_j/\lambda_i}|^2< \theta$. This condition implies $|\lambda_i - \mu_j|^2 \leq \theta\mu_j\lambda_i \leq \theta(1+\eta)^{-3} \tau_k^2$. Since $\tau_k$ is chosen uniformly in an interval of length $\tau_k \eta (1+\eta)^{-1}$, the expected fraction of pairs $(i,j)$ such that such that $|\sqrt{\lambda_i/\mu_j}-\sqrt{\mu_j/\lambda_i}|^2< \theta$ and $\lambda_i \leq \tau_k \leq \mu_j$ is at most $O(\sqrt{\theta}/\eta)$. Hence, on expectation over the choice of the $\tau_k$ we have
$$
\sum_{\substack{k\neq k',\,i\in S_k,j\in T_{k'}\\ |\sqrt{\lambda_i/\mu_j}-\sqrt{\mu_j/\lambda_i}|^2< \theta}} \lambda_i\mu_j |\bra{u_i}v_j\rangle|^2\,\leq\,  O(\sqrt{\theta}\eta^{-1})\sum_{i,j} \lambda_i\mu_j |\bra{u_i}v_j\rangle|^2\,=\, O(\sqrt{\theta}\eta^{-1}).
$$
Choosing $\theta = (\delta\eta)^{2/3}$, we obtain that~\eqref{eq:pq-close-00} holds, on expectation over the choice of the $\tau_k$, with a right-hand side of $1-O(\delta^{1/3}\eta^{-2/3})$. (The condition that $\theta \ll \eta$ is equivalent to $\delta \ll \eta^{1/3}$, which we may assume holds without loss of generality, as otherwise the bound in the claim is trivial.) The left-hand side is at most $1$, and applying Markov's inequality proves the claim.
\end{proof}

Our last claim analyzes the outcome of the sampling procedure, proving the lemma.

\begin{claim}\label{claim:correlated}
Let $\ket{\psi}$, $\ket{\varphi}$ be such that $\|\ket{\psi}-\ket{\varphi}\|^2 \leq \delta$, and set $\eta=\delta^{1/4}$. With probability at least $1-O(\delta^{1/12})$, the sampling procedure described above terminates with Alice and Bob in a shared state $\ket{\xi}$ such that $\|\ket{\xi}-\ket{\psi}\|^2=O(\delta^{1/12})$.
\end{claim}

\begin{proof}
Suppose first that~\eqref{eq:pq-close-00} holds and that Alice and Bob both proceed to the step 6 synchronously. In that case, at the end of the procedure their joint state is
$$ \ket{\xi} \,:=\, C'' \sum_k \tau_k \sum_{i\in S_k,j\in T_k} \bra{u_i}v_j\rangle \ket{u_i}\ket{v_j},$$
where the normalization constant $C''$ satisfies
$$(C'')^{-2} \,=\,\sum_k \tau_k^2 \sum_{i\in S_k,j\in T_k} | \bra{u_i}v_j\rangle|^2 = \sum_k \tau_k^2 \Tr(P_kQ_k) = 1- O(\eta+\delta^{1/3}\eta^{-2/3})$$
 by Claim~\ref{claim:pq-close}.
We can thus evaluate the overlap of $\ket{\xi}$ with $\ket{\Phi}$ as
\begin{align*}
\bra{\xi}\Phi\rangle &\geq \sum_k \tau_k^2 \,\sum_{i\in S_k,j\in T_k} |\bra{u_i}v_j\rangle|^2 - O(\delta^{1/3}\eta^{-2/3}) \\
&=  1-O(\delta^{1/6}\eta^{-1/3}),
\end{align*}
where for the first equality we used orthogonality of the $\ket{u_i}$, and the last again follows from Claim~\ref{claim:pq-close}.

Next we compute the probability that in step 5 Alice and Bob both obtain the first outcome of their respective POVM in the same iteration. The probability that Alice alone obtains a successful outcome is $\sum_k \tau_k^2 s_k /(d\sum_k \tau_k^2) = (1+\Theta(\eta))(d\sum_k \tau_k^2)^{-1}$ by~\eqref{eq:corr-3}. The same holds for Bob. With probability at least $1-\delta^2$, both of them obtain a successful outcome before the number $N$ of copies of $\ket{\xi_0}$ runs out. Moreover,
the probability that they simultaneously obtain the first outcome is
$$\big(d\sum_k \tau_k^2\big)^{-1}\sum_k \tau_k^2 \Tr(P_kQ_k) \,\geq\, \big( 1- O(\delta^{1/6}\eta^{-1/3})\big)\big(d\sum_k \tau_k^2\big)^{-1}$$
by Claim~\ref{claim:pq-close}. Hence the probability that they simultaneously proceed to the third step of the protocol is at least $1- O(\delta^{1/6}\eta^{-1/3})$. Choosing $\eta = \delta^{1/4}$ proves the lemma.
\end{proof}

\end{proof}

\bibliographystyle{alphaabbrvprelim}

\bibliography{parrep}

\end{document}